\documentclass[11pt]{article} % use larger type; default would be 10pt

\usepackage[utf8]{inputenc} % set input encoding (not needed with XeLaTeX)

\usepackage{geometry} % to change the page dimensions
\usepackage{amsmath}
\usepackage{amssymb}
\usepackage{amsthm, cleveref}
\usepackage{vector}
\input fullpage.sty
\usepackage{graphicx} % support the \includegraphics command and options
\usepackage{tikz}

\usepackage{booktabs} % for much better looking tables
\usepackage{array} % for better arrays (eg matrices) in maths
\usepackage{paralist} % very flexible & customisable lists (eg. enumerate/itemize, etc.)
\usepackage{verbatim} % adds environment for commenting out blocks of text & for better verbatim
\usepackage{subfig} % make it possible to include more than one captioned figure/table in a single float
\usepackage{xspace}
\usepackage{sectsty}
\allsectionsfont{\sffamily\mdseries\upshape} % (See the fntguide.pdf for font help)
\newtheorem{theorem}{Theorem}[section]

\newtheorem{definition}[theorem]{Definition}

\newtheorem{lemma}[theorem]{Lemma}

\newtheorem{proposition}[theorem]{Proposition}

\newtheorem{corollary}[theorem]{Corollary}

\newtheorem{claim}[theorem]{Claim}

\newtheorem{remk}[theorem]{Remark}

\usepackage[nottoc,notlof,notlot]{tocbibind} % Put the bibliography in the ToC
\usepackage[titles,subfigure]{tocloft} % Alter the style of the Table of Contents
\usepackage{tikz}
\usepackage[affil-it]{authblk}

\newcommand{\grant}[1]{}
\newcommand{\aaron}[1]{}
\newcommand{\paolo}[1]{}
\newcommand{\GI}[0]{{\textsc{GraphIsomorphism}}\xspace}

\newcommand{\F}[0]{{\mathbb{F}}}
\newcommand{\CFI}[0]{{Cai-F\"urer-Immerman }}
\title{Graph Isomorphism and the Lasserre Hierarchy}
\author{Paolo Codenotti}
\affil{Google Inc.}

\author{Grant Schoenebeck}
\affil{University of Michigan}

\author{Aaron Snook}
\affil{University of Michigan}

\begin{document}
\maketitle

\begin{abstract}
In this paper we show lower bounds for a certain large class of algorithms solving the Graph Isomorphism problem, even on expander graph instances.  Spielman~\cite{Spielman96} shows an algorithm for isomorphism of strongly regular expander graphs that runs in time $\exp\{\tilde{O}(n^\frac{1}{3})\}$ (this bound was recently improved to $\exp\{\tilde{O}(n^{\frac{1}{5}})\}$~\cite{BabaiCSSW13}).  It has since been an open question to remove the requirement that the graph be strongly regular.  Recent algorithmic results show that for many problems the Lasserre hierarchy works surprisingly well when the underlying graph has expansion properties.  Moreover, recent work of Atserias and Maneva~\cite{AtseriasM12} shows that $k$ rounds of the Lasserre hierarchy is a generalization of the $k$-dimensional Weisfeiler-Lehman algorithm for Graph Isomorphism.  These two facts combined make the Lasserre hierarchy a good candidate for solving graph isomorphism on expander graphs.  Our main result rules out this promising direction by showing that even $\Omega(n)$ rounds of the Lasserre semidefinite program hierarchy fail to solve the Graph Isomorphism problem even on expander graphs.
\end{abstract}

\newpage

\setcounter{page}{1}
\section{Introduction}
We analyze the Lasserre relaxations of the quadratic program for the \GI problem. In particular we show that a linear number of levels of the Lasserre hiearchy are required to distinguish certain classes of pairs of non-isomorphic graphs. This implies that a large class of semidefinite programs fails to provide a sub-exponential time algorithm for \GI .   Our construction uses the construction of Cai, F\"urer, and Immerman~\cite{CaiFI89}, which was used to show that \GI cannot be solved in polynomial time by the $k$-dimensional Weisfeler-Lehman process ($k$-WL from now on, see Appendix~\ref{sec:Partition-Refinement} for more details), an algorithm subsuming a wide class of combinatorial algorithms. Our result therefore implies that, in the worst case, the Lasserre hierarchy does not out-perform these combinatorial algorithms.

The main motivation to study the \GI problem is its unique complexity-theoretic status. \GI is in NP $\cap$ coAM (\cite{GolreichMW91} cf.~\cite{BabaiM88}), and hence if \GI were NP-Complete, then the polynomial hierarchy would collapse to the second level (\cite{BoppanaHZ87}, cf.~\cite{BabaiM88}). On the other hand, the best known running-time for an algorithm for \GI is $\exp(O(\sqrt{n\log(n)}))$~\cite{BabaiL83}. \GI is the only natural problem with this status.  As such, we hope that studying the Lasserre of \GI  can also give unique insights into the power of semidefinite programs.

Recent algorithmic results show that when the underlying graph has expansion properties the Lasserre hierarchy works surprisingly well for problems including {\textsc{UniqueGames}}\cite{AroraKKSTV08,BarakRS11,GuruswamiS11},
 {\textsc{GraphColoring}}~\cite{AroraG11}, {\textsc{SparsestCut}}~\cite{GuruswamiS11,GuruswamiS13}, {\textsc{MinBisection}}, {\textsc{EdgeExpansion}}, and {\textsc{SmallSetExpansion}}~\cite{GuruswamiS11}.  The special case of expanders is particularly interesting for the \GI problem since expansion properties have been leveraged to obtain better bounds for the special class of strongly regular graphs.\footnote{A graph is \emph{strongly regular} with parameters $(k, \lambda, \mu)$ if it is regular of degree $k$, and every pair of vertices $(u, v)$ has $\lambda$ common neighbors if $u$ and $v$ are adjacent, and $\mu$ common neighbors otherwise.} This is a class of highly structured graphs, which had been thought to be (and might still be) a major bottleneck to the \GI problem.
Using the $k$-WL algorithm and bounds on the expansion of the graphs, Spielman was able to obtain an algorithm with running time $\exp(O(n^{1/3}\text{poly}\log(n)))$~\cite{Spielman96}. More recently this has been improved to $\exp(O(n^{1/5}\text{poly}\log(n)))$~\cite{BabaiCSSW13}. Both algorithms critically make use of expansion properties of certain classes of strongly regular graphs, raising the question of whether there are faster algorithms for isomorphism of expander graphs. We show that the graphs in our construction are expanders, and hence our result implies that the Lasserre hierarchy cannot solve \GI in sub-exponential time even in this special case.

Our result is an extension of the work of Atserias and Maneva, who showed that $k$-WL is equivalent to $k\pm 1$ rounds of the Sherali-Adams hierarchy~\cite{AtseriasM12} (a large class of linear programs).  This extension is not trivial as for both {\textsc{Maxcut}} and {\textsc{VertexCover}} linear program hierarchies have been shown to fail on expanding graphs where very basic semidefinite programs succeed~\cite{STT-07b,STT-2007a,CharikarMM-07}.

The main technical ingredients of the proof consist in writing the isomorphism problem as a set of linear constraints over $\mathbb{F}_2^n$, and then relating the cut-width of the graphs to the width of resolution proofs for the constraints. Finally we construct vectors which are solutions to the Lasserre relaxations using Schonebeck's construction for partial assignments to formulas with no small width resolution proof~\cite{S08}.

Please see Appendix~\ref{appendix:background} for more background on  \GI.

\paragraph{Independent Results}
Similar results were shown independently by O'Donnell, Wright, Wu, and Zhou~\cite{ODonnellWWZ2014}.  Through a construction also based on \CFI~graphs, they also show that $\Omega(n)$ rounds of Lasserre fail to solve \GI. Their reduction has the additional property which allows them to extend the results to the ``Robust Graph Isomorphism" problem, where they show that Lasserre also fails to distinguish graphs that are ``far" from being isomorphic -- any permutation violates many edges.

%
%\subsection{The Graph Isomorphism Problem}\label{ssec:GI}

%\paolo{Remove the next 2 paragraphs, already said in the next section}
%An isomorphism of two graphs $G$ and $H$ is a bijection $f:V(G) \to V(H)$ that preserves edge relationships (meaning $(u, v) \in E(G)$ if and only if $(f(u), f(v)) \in E(H))$. Given two input graphs, $G$ and $H$, the \GI problem asks whether or not there exists an isomorphism between them.
%\paolo{Do we need the alternate formulation here? we'll talk about it when we discuss the semidefinite programming formulation}.
%A useful equivalent formulation of the \GI problem is: given (adjacency) matrices $A$ and $B$ (of graphs $G$ and $H$), does there exist a permutation matrix $X$ such that $X^{-1}H X = G$?

%It will be useful to consider a generalization of \GI to vertex colored graphs. Two vertex colored graphs $G$ and $H$ are said to be isomorphic if there is an isomorphism that preserves the colors. More precisely, if the color of vertex $v$ is given by $c(v)$, an isomorphism $f:V(G)\to V(H)$ must additionally satisfy $c(f(v)) = c(v)$ for all $v \in V(G)$. It has been shown that colored graph isomorphism and graph isomorphism are polynomial-time equivalent \cite{coloredG}.

\section{Background and Notation}
%\subsection{Definitions and Notation}
\paragraph{Functions}
We use the notation $[k]$ to denote the set $\{1, \ldots, k\}$.  Let $\mathcal{P}(S)$ denote the powerset of a set $S$.
Given a partial function $f: S \rightarrow T$, we define $dom(f)$ to be $\{x \in S \mid f(x) \text{ is defined }\}$
For a function $f: S \rightarrow T$, we denote the range of $f$ on $S' \subseteq S$ by $f(S) \subseteq T$. We denote the preimage of f
on a set $T' \subseteq T$ as $f^{-1}(T')$. (The preimage of a set $T' \subseteq T$ is the set $\{ x \in S \mid f(x) \in T'\}$).
Given a function $f : S \rightarrow \{0, 1\}$, with $S$ a finite set, we define the \emph{parity} of $f$ to be the parity of $|\{ x \in S \mid f(x) = 1\}|$.
We will denote a function as ${\bf 0}$ (the zero function) if it maps all its inputs to 0.  A function $f: \{0, 1\}^n \rightarrow \{0, 1\}$ is a \emph{$k$-junta} if it only depends on $k$ variables.  That is there exist $k$ integers $i_1, \ldots, i_k$ and a function $g: \{0, 1\}^k \rightarrow \{0, 1\}$ such that for all $x \in \{0, 1\}^n$ we have $f(x_1, x_2, \ldots, x_n) = g(x_{i_1}, x_{i_2}, \ldots, x_{i_k})$.  \begin{definition} \label{def:h} For a partial function $\alpha: [n] \rightarrow \{0, 1\}$, we define the indicator function $h_\alpha: \{0, 1\}^n \rightarrow \{0, 1\}$ as \[
h_\alpha(x) =
\begin{cases}
1 & \text{if for all } i \in [n] \text{ such that } \alpha(i) \text{ is defined, } x_i = \alpha(i)\\
0 & \text{otherwise}
\end{cases}\]  Note that $h_\alpha$ is a $|dom(\alpha)|$-junta.
\end{definition}

\paragraph{Graphs}
An undirected \emph{graph} $G$ is a tuple $(V, E)$ such that $V$ is a finite set of \emph{vertices}, and $E \subseteq {V \choose 2}$. We
refer to the elements of $E$ as \emph{edges}.
Given a graph $G$, we call $V(G)$ the set of vertices of $G$, and $E(G)$ the set of edges of $G$. We will represent an edge of $E$ as a tuple of its two vertices. Given a graph $G$ and $I, J \subseteq V(G)$, let $E(I, J) \subseteq E(G) = \{(i, j) \mid i \in I, j \in J\}$.
Also, we define % the function $\Gamma: V(G) \rightarrow \mathcal{P}(V(G))$,
the set of \emph{neighbors} of $v$,  as $\Gamma(v) = \{u \mid (v, u) \in E(G)\}$.
Given a set $S \subseteq G$, let the \emph{induced subgraph of G on S} be the graph $(S, E')$ where $E' = \{(u,v) \in E(G) \mid u, v \in S\}$.
An undirected \emph{colored graph} $G$ is a tuple $(V, E, c)$, where $V$ and $E$ are the same as in uncolored graphs, and $c : V(G) \rightarrow C$ is a function
mapping vertices of G to a finite set of \emph{colors}. Also, given a colored graph $G$ and a vertex i, we define $C(i)$ as the set of vertices with the same color as i.
The \emph{expansion} of a set $S \subseteq V(G)$ is equal to $Ex(S) = \frac{E(S, V(G) \setminus S)}{\min\{|S|, |V(G) \setminus S|\}}$.
The expansion of a graph is $Ex(G) = \min_{S \subseteq V(G)} Ex(S)$.  \begin{theorem}\label{thm:three-reg-exp}\cite{Ellis}  A random 3-regular graph has expansion greater than $0.54$ with probability $1 - o(1)$.\end{theorem}

The \emph{cutwidth} of a graph $G$ is defined to be $$CW(G) = \min_{\pi} \max_i E(\pi([i]), V(G) \setminus \pi([i]))$$ where $\pi: V(G) \rightarrow [n]$ is a permutation that orders the vertices; so that $\pi([i])$ denotes the first $i$ variables in the ordering.  We define the width of a graph $G$ as follows: $$W(G) = \max_{\Omega} \min_{(S_1, S_2) \in \partial\Omega} \max_{i \in \{1, 2\}} E(S_i, V(G) \setminus S_i)$$ where $\Omega$ ranges over monotone sets of $\mathcal{P}(V(G))$ (i.e. $S_1 \subseteq S_2$ and $S_2 \in \Omega \Rightarrow S_1 \in \Omega$) such that $\emptyset \in \Omega$ and $V(G) \not\in \Omega$; and $(S_1, S_2) \in \partial\Omega$ if $S_1 \in \Omega$, $S_2 \not\in \Omega$ and $S_1 = S_2 \cup \{i\}$ for some element $i$.
\begin{theorem} \label{theorem:CW(G)=W(G)} \cite{MontanariS09} For any graph $G$, $CW(G) = W(G)$\end{theorem}
It will be useful to state the results in terms of the cut-width.  However, the proofs will use the graph width.  This theorem tells us they are the same.

\begin{definition}  \label{def:k,t-stretch}
A graph $H$ is said to be a \emph{$(k, t)$-stretching} of $G$ if there exists an injective function $g: V(G) \rightarrow V(H)$  such that:
\begin{itemize}
\item Each vertex $v \in V(H) \setminus g(V(G))$ has degree 2.
\item $(u, v) \in E(G)$ iff there exists a unique path $P(u, v) = (g(u) = p(u, v)_0, p(u, v)_1, \ldots, p(u, v)_{\ell} = g(v))$ from $g(u)$ to $g(v)$ in $H$ such that for $1 \leq i \leq \ell -1$ we have $p(u, v)_i \not\in g(V(G))$.
     %This is, the unique path from $f(u)$ to $f(v)$ in $H$ using only vertices of $H \setminus f(V(G))$ is of length at most $k$ and does not use a vertex in $f(V(G))$ except for $f(u)$ and $f(v)$.
\item For every simple path $p =  (p_0, p_1, \ldots, p_{\ell})$ in $H$ where $p_0, p_{\ell} \in g(V(G))$ and $p_i \not\in g(V(G))$ for all $1 \leq i \leq \ell -1$,  the length $\ell$ of this path is at most $k$.
\item For every vertex $v \in V(G)$ the number of vertices in $\bigcup_{u \in \Gamma_G(v)} P(v, u)) \setminus g(V(G))$ is at most $t$.
\end{itemize}
We say that $g$ \emph{witnesses} the fact that $H$ is as $(k, t)$ stretching of $G$.
%We denote the induced subgraph of $H$ formed by taking this path as $P(u, v) = P(u, v)_0, P(u, v)_1, \ldots, P(u, v)_{\ell}$; and will index into the vertices of this path with subscripts (so $P(u, v)_1$ is the vertex that u is adjacent to in $P(u, v)$).\\
\end{definition}

\paragraph{Isomorphisms and Permutations}
Two (undirected) graphs $G$ and $H$ are said to be \emph{isomorphic} if and only if there exists a bijective mapping $\pi: V(G) \rightarrow V(H)$ such that  $(u, v) \in E(G) \iff (\pi(u), \pi(v)) \in E(H)$. If $G$ and $H$ are colored, we further require that $\forall v \in V(G)$, $c(v) = c(\pi(v))$.
We then say that $\pi$ is an \emph{isomorphism} between G and H.
For (not necessarily isomorphic) $G$ and $H$, we define a \emph{partial isomorphism} $\sigma: V(G) \rightarrow V(H)$ as an injective partial mapping between $V(G)$ and $V(H)$ such that $\sigma$ is an isomorphism between the induced subgraph of $G$ on $dom(\sigma)$ and the induced subgraph of $H$ on $\sigma(dom(\sigma))$.
We denote the partial isomorphism $\emptyset$ as the partial isomorphism with $dom(\sigma) = \emptyset$. We denote by $i \rightarrow i'$  the partial isomorphism that
only maps $i \in V(G)$ to $i' \in V(H)$.

If $\sigma_1$ and $\sigma_2$ are partial permutations, we say $\sigma_1$ and $\sigma_2$ are \emph{consistent} if 1)
$i \in dom(\sigma_1) \cap dom(\sigma_2) \Rightarrow \sigma_1(i) = \sigma_2(i)$ and 2)  $\sigma_1(i) = \sigma_2(j) \Rightarrow i = j$.

Let $PP(G)$ be the set of partial permutations.  Let $P(G) = PP(G) \cup \{\bot\}$.

We define a function $\wedge: P(G) \times P(G) \rightarrow P(G)$ as follows:
$ \sigma_1 \wedge \sigma_2 = \bot$ if 1) $\sigma_1 = \bot$ 2)  $\sigma_2 = \bot$ or 3)  $\sigma_1$ and $\sigma_2$ are not consistent.
Otherwise $$\sigma_1 \wedge \sigma_2(i) = \left\{ \begin{array}{cc} \sigma_1(i)   & i \in dom(\sigma_1) \\
                                                            \sigma_2(i)   & i \in dom(\sigma_2) \setminus  dom(\sigma_1)  \\
                                                            \mbox{undefined}  &  i \not\in dom(\sigma_2) \cup  dom(\sigma_1)
                                                    \end{array}\right.$$

\paragraph{Linear constraints}
A \emph{linear constraint} on $k$ variables over $\mathbb{F}_2$ is a constraint of the form $\left(\bigoplus_{i \in S} x_i = b \right)$ for some set $S \subseteq [n]$ and $b \in \{0, 1\}$. % $\sum_{i \in [1..k]} a_i x_i = b$, where
%$a_i \in \mathbb{F}$.
%
%A \emph{linear constraint} on $k$ variables over $\mathbb{F}$ as being a constraint $\sum_{i \in [1..k]} a_i x_i = b$, where
%$a_i \in \mathbb{F}$. In this paper, $\mathbb{F}$ will always be $\mathbb{F}_2^n$ and so we will write a linear constraint as
%$\left(\bigoplus_{i \in S} x_i = b_i\right)$ for some set $S$. The constraints we will deal with in this paper can also be viewed as parity constraints
%over the variables in S.
Given two linear constraints $X_1 = \left(\bigoplus_{i \in I} x_i = b_i\right)$ and $X_2 = \left(\bigoplus_{j \in J} x_i = b_j\right)$, we
abuse notation so that $X_1 \oplus X_2$ denotes the constraint $\left(\bigoplus_{i \in I \triangle J} x_i = b_i \oplus b_j\right)$.  Similarly we say that $X_1 = X_2$ if $I = J$ and $b_i = b_j$.

\paragraph{Fourier Analysis}
%First, we will define the character of I, a function of interest in our analysis.
%\aaron{I'm not sure we're still interested in this...I should review the paper and see if still needed.}

%Given $I \in [n]$, we define the function $\chi_I : \{0, 1\}^n \rightarrow \{0, 1\}$ as \[
%\chi_I({\bf x}) = {\bigoplus_{i \in I} x_i}\]
%\aaron{Not sure where to put the following either. Not strictly FA, but it's always used in context of FA.}

Those interested in a broader view of Fourier Analysis are referred to \cite{ODonnell2013}.
Given a function $f: \{0, 1\}^n \rightarrow \{0, 1\}$, we define $\hat{f}: \mathcal{P}([n]) \rightarrow [-1, 1]$
by \[
\hat{f}(I) = \mathbb{E}_x[f(x)\cdot(-1)^{\bigoplus_{i \in I} x_i}]
\]

We will use the following facts:
\begin{itemize}
\item For any functions $f, g : \{0, 1\}^n \rightarrow \{0, 1\}$, we have $\widehat{f+g}(I) = \hat{f}(I) + \hat{g}(I)$
\item For any functions $f, g : \{0, 1\}^n \rightarrow \{0, 1\}$, we have $\widehat{f\cdot g}(I) = \sum_{J \subseteq [n]} \hat{f}(J)\hat{g}(I \triangle J)$
\item Let $f:\{0, 1\}^n \rightarrow \{0, 1\}$ be a $k$-junta, then for any $|I| > k$, $\widehat{f}(I) = 0$
\item For partial function $\alpha: [n] \rightarrow \{0, 1\}$ defined on $I \subseteq [n]$, for $J \subseteq I$, $\hat{h}_{\alpha}(J) = (-1)^{\bigoplus_{i \in J} \alpha(i)} \frac{1}{2^{|I|}}$, and for $J \not\subseteq I$,  $\hat{h}_I(J) = 0$, where $h_\alpha$ is defined as in Definition~\ref{def:h}. \\
%$\hat{\chi_\emptyset}(\emptyset) = \hat{\bf 1}(\emptyset) = 1$, $\hat{\bf 1}(I) = 0$ if $I \neq \emptyset$.
\end{itemize}

\paragraph{Probability}

If $\{A_1, A_2...\}$ is a sequence of events under distinct probability spaces, we say that $A_n$ occurs with \emph{high probability}
if $\mathbb{P}(A_n) \rightarrow 1$ as $n \rightarrow \infty$.

%\paragraph{K-stretching of a graph}
%A  k-stretching [strict k-stretching] of a graph G can be thought of as a graph where at most [exactly] k vertices are ``inserted'' as midpoints to each edge
%in G. Formally,
%a graph $H$ is said to be a k-stretching [strict k-stretching] of $G$ if there is an injective function $f: V(G) \rightarrow V(H)$ such that:
%\begin{itemize}
%\item Each vertex $v \in V(H) \setminus f(V(G))$ has degree 2
%\item $(u, v) \in E(G)$ iff there exists a path of length at most[exactly] k+1 from $f(u)$ to $f(v)$ that does not use a vertex in $f(V(G))$ except for $f(u)$ and $f(v)$.
%\end{itemize}
%We denote the induced subgraph of H formed by taking the unique path of length k+1 from $f(u)$ to $f(v)$ as $P(u, v)$;
%we will index into the vertices of this path with subscripts, so $P(u, v) = \{P(u, v)_0 = u, P(u, v)_1...P(u, v)_{k+1} = v\}$.
%%\aaron{Actually, it seems we need the Fourier analysis of AND functions, not characters, later on.}
%
%%Intuitively, $f$ specifies whether or not the edge $(u, v)$ is ``twisted" or not. We have the following notational shortcuts: ${\bf 0}$ denotes
%%the zero function, and ${\bf (u, v)}$ is the function that is 1 on $(u, v)$ and 0 elsewhere.\\

\subsection{Quadratic Program and Lasserre Relatation of \GI}

The linear algebraic formulation of Graph Isomorphism is: given graphs G and H each with $n$ vertices, represented by adjacency matrices A and B respectively,
does there exist a permutation matrix such that $X^{\top}BX = A$? If we label the vertices of both $G$ and $H$ with the integers from 1 to n, this can be represented by the quadratic program
\begin{equation} \label{eq:las}
\forall i, j \:\:\: \sum_{i', j' \in [1..n]} x_{i\rightarrow i'}x_{j' \rightarrow j}B_{i'j'} = A_{ij}
\end{equation}
\begin{equation*}
\forall i, i'  \:\:\:  x_{i \rightarrow i'}(1 - x_{i \rightarrow i'}) = 0
\end{equation*}
\begin{equation*}
\forall i  \:\:\: \sum_{i' \in [1..n]} x_{i \rightarrow i'} = 1
\end{equation*}
\begin{equation*}
\forall i'  \:\:\:  \sum_{i \in [1..n]} x_{i \rightarrow i'} = 1
\end{equation*}
where $x_{i \rightarrow i'}$ is an indicator variable that is 1 if vertex i in $G$ maps to vertex i' in $H$ and 0 otherwise.
We consider the Lasserre hierarchy of this program, which relaxes the scalar variables to vectors instead. These are denoted by $v_{\sigma}$ where
$\sigma$ is a partial isomorphism from $G$ to $H$; a simple example of such an $\sigma$ is the single mapping $i \rightarrow i'$.

We consider the Lasserre hierarchy of relaxations of this quadratic program.
We define $\Sigma_r$ as the set of all partial isomorphisms from $G$ to $H$ with domain at most r.
The constraints on the $r$th level are as follows:

%I can see it being a potential question - why are the Lasserre constraints like this?

\begin{subequations}
\begin{equation} \label{eq:l1}
||v_\emptyset|| = 1
\end{equation}
\begin{equation} \label{eq:l2}
\forall i, j \sum_{i', j'} \langle v_{i \rightarrow i'}, v_{j \rightarrow j'} \rangle B_{i'j'}= A_{ij}
\end{equation}
\begin{equation} \label{eq:l3}
\forall \sigma_1, \sigma_2 \in \Sigma_r \text{ s.t. } \sigma_1 \wedge \sigma_2 = \sigma_1' \wedge \sigma_2', \:\:\:  \langle v_{\sigma_1}, v_{\sigma_2} \rangle = \langle v_{\sigma'_1}, v_{\sigma'_2} \rangle
\end{equation}
\begin{equation} \label{eq:l4}
\forall \sigma \in \Sigma_r, \forall i \in V(G),  \:\:\:  v_\sigma = \sum_{i' \in V(H)} v_{\sigma \wedge (i \rightarrow i')}
\end{equation}
\begin{equation} \label{eq:l5}
\forall \sigma \in \Sigma_r, \forall i' \in V(G),  \:\:\:  v_\sigma = \sum_{i \in V(H)} v_{\sigma \wedge (i \rightarrow i')}
\end{equation}
\end{subequations}

It can be verified that if $\pi_1...\pi_k$ are each isomorphisms from G to H, $p_1, \ldots, p_k$ is a probability distribution (meaning that $\forall i\text{, }
p_1 \geq 0$
and $\sum_{i=1}^k p_i = 1$), then the following choice of vectors will satisfy \eqref{eq:l1}-\eqref{eq:l5}: \\
For $\sigma$ a partial isomorphism, \[
v_{\sigma}(i) = \begin{cases}
\sqrt{p_i} & \text{ if for each } j \text{ with } \sigma(j) \text{ defined }, \pi_i(j) = \sigma(j)'\\
0 & \text{ otherwise }
\end{cases}
\]
This shows that indeed, the Lasserre SDP is a relaxation of the quadratic program.

\subsection{The \CFI~Graphs} %TODO: should this be a preliminary?

%The Cai-Furer-Immerman Graphs were first used in \cite{CaiFI89}, a seminal part of the body of work on Graph Isomorphism. It showed that the canonical combinatorial attempt at solving Graph Isomorphism, the Weisfeiler-Lehman method, fails to solve GI in polynomial time. This was a surprising result at the time.

In this section we show how to apply the \CFI~gadget to a 3-regular graph.  Our main result will be based on deciding \GI~for these types of graphs.

%For the purposes of this section, we fix $G$ as a (randomly chosen) 3-regular simple graph on $n$ vertices. We assume that $G$ is connected and
%is an expander with constant $c \geq 0.18$, which happens with with high probability \cite{3Reg}.\\

The following definition presents the gadgets which will be used to construct the \CFI~graphs.

\begin{definition}   Given a vertex $v$ with degree 3 and neighbors $u_1, u_2, u_3$. \footnote{For the sake of exposition, we fix an ordering of the three neighbors of each vertex.}  We define the colored graph $CFI(v)$ as follows, from \cite{CaiFI89}:
\begin{itemize}
\item $V(CFI(v)) = M(v) \cup E(v)$
\item $M(v)$, or the ``middle vertices'' of $CFI(v)$, is a set of 4 vertices each labeled with one of the 4 different even parity functions from $\Gamma(v) \rightarrow \{0, 1\}$.  We will denote these vertices $v_{b_{u_1}, b_{u_2}, b_{u_3}}$ where  $b_{u_i} \in \{0, 1\}$ for $i = \{1, 2, 3\}$ and $b_{u_1} \oplus b_{u_2} \oplus b_{u_3} = 0$.  The color of these vertices is $c(v)$.
\item $E(v)$, or the ``edge vertices'' of $CFI(v)$, is a set of 3 pairs of vertices: $(v, u_i)_b$ for a different $i \in \{1, 2, 3\}$ and $b \in \{0, 1\}$.  For each $i \in \{1, 2, 3\}$, the pair $(v, u_i)_0$, $(v, u_i)_0$ is colored $(c(v), u_i)$ (a new color, determined by $c(v)$ and $u_i$).
\item $E(CFI(v)) = \{((v, u_i)_b, v_{b_{u_1}, b_{u_2}, b_{u_3}}) \mid i \in \{1, 2, 3\}, b = b_{u_i}\}$; stated differently, $v_{b_{u_1}, b_{u_2}, b_{u_3}}$ is connected to $(v, u_1)_{b_{u_1}}$,
$(v, u_2)_{b_{u_2}}$, and $(v, u_3)_{b_{u_3}}$.
\end{itemize}
\end{definition}

The below figure illustrates $CFI(v)$, the graph gadget, zoomed in at a particular vertex.

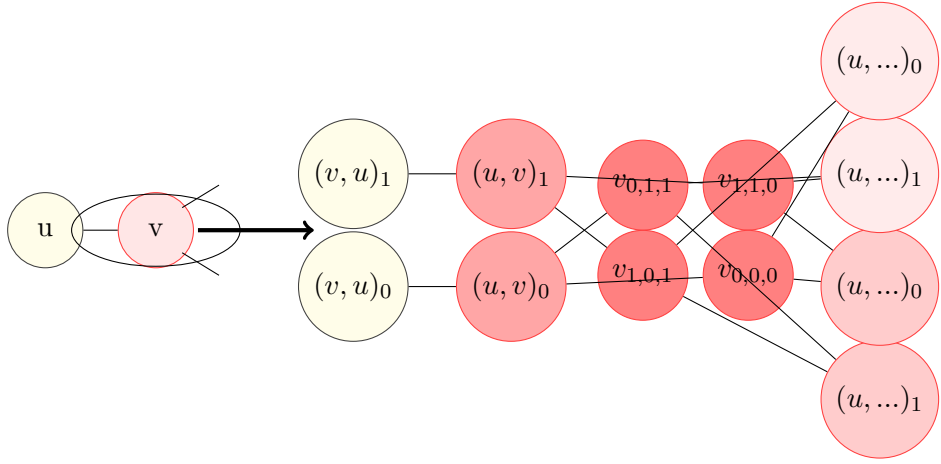
\vspace{10mm}

\begin{figure}[h] \label{fig:CFI-basic}
\centering
\begin{tikzpicture}[xscale=0.7, yscale=0.6]

   \def\myellipse{(0, 3.5) ellipse (1.6cm and 0.8cm)};

   \tikzstyle{mid}=[circle,draw=red!75,fill=red!50,minimum size=10mm]
   \tikzstyle{e1}=[circle,draw=red!75,fill=red!35,minimum size=7mm]
   \tikzstyle{e2}=[circle,draw=red!75,fill=red!20,minimum size=7mm]
   \tikzstyle{e3}=[circle,draw=red!75,fill=red!8,minimum size=7mm]
   \tikzstyle{e4}=[circle,draw=black!75,fill=yellow!10,minimum size=7mm]
   [scale=.4]
  \begin{scope}
  \node[circle,draw=black!75,fill=yellow!10,minimum size=10mm] (u) at (-2.1, 3.5) {u};
  \node[circle,draw=red!75,fill=red!10,minimum size=10mm] (v) at (0, 3.5) {v};
  \end{scope}
  \begin{scope}[xshift=.25cm]
  \node [mid] (m12) at (9,2.5)  {$v_{1, 0, 1}$};
  \node [mid] (m13) at (11,4.5)  {$v_{1, 1, 0}$};
  \node [mid] (m23) at (9,4.5)  {$v_{0, 1, 1}$};
  \node [mid] (m) at (11,2.5)  {$v_{0, 0, 0}$};
  \node [e4] (u1) at (3.5,4.75) {$(v, u)_1$};
  \node [e4] (u2) at (3.5,2.25) {$(v, u)_0$};
  \node [e1] (t1) at (6.5,4.75)  {$(u, v)_1$};
  \node [e1] (f1) at (6.5,2.25) {$(u, v)_0$};
  \node [e2] (t2) at (13.5,-0.25)  {$(u,...)_1$};
  \node [e2] (f2) at (13.5,2.25) {$(u,...)_0$};
  \node [e3] (t3) at (13.5,4.75)  {$(u,...)_1$};
  \node [e3] (f3) at (13.5,7.25) {$(u,...)_0$};
  \foreach \from/\to in {m/f1,m/f2,m/f3,m12/t1,m12/t2,m12/f3,m13/t1,m13/f2,m13/t3,m23/f1,m23/t2,m23/t3, u1/t1, u2/f1}
    \draw (\from) -- (\to);
   \end{scope}

  \draw \myellipse ;
  \draw  (u) -- (v);
  \draw (0.5, 4.0) -- (1.2, 4.5);
  \draw (0.5, 3.0) -- (1.2, 2.5);
  \draw [ultra thick, ->] (0.8, 3.5) -- (3.0, 3.5);

\end{tikzpicture}
%\caption{X(v) and the edge between X(u) and X(v) in $X_f(G)$ if $f(u, v) = 0$}
\end{figure}

%\vspace{10cm}

\begin{remk}  \label{obs:color restricted}
Note that the colors of $CFI(v)$ ensure that any automorphisms of $CFI(v)$ must map $M(u)$ to $M(u)$ and must map each exterior pair to itself.
\end{remk}
\paragraph{}
The following claim helps to elucidate an important aspect of this gadget:
\begin{claim} \label{eq:unique auto}
If $\Gamma(v) = \{u_1, u_2, u_3\}$, then for each pair of vertices $v_{b_{u_1}, b_{u_2}, b_{u_3}}, v_{b'_{u_1}, b'_{u_2}, b'_{u_3}}
\in M(v)$, there is a unique automorphism $f$ of $CFI(v)$ such that $f(v_{b_{u_1}, b_{u_2}, b_{u_3}}) = v_{b'_{u_1},b'_{u_2}, b'_{u_3}}$.
\end{claim}

\begin{proof} %\grant{check this and make sure it is clear enough}
Let $f$ be an automorphism of $CFI(v)$, and let  $v_{b_{u_1}, b_{u_2}, b_{u_3}}, v_{b'_{u_1}, b'_{u_2}, b'_{u_3}}
\in M(v)$ be such that \[f(v_{b_{u_1}, b_{u_2}, b_{u_3}}) = v_{b'_{u_1}, b'_{u_2}, b'_{u_3}}\]
Since $\forall i \in \{1, 2, 3\} (v_{b_{u_1}, b_{u_2}, b_{u_3}},(v, u_i)_{b_{u_i}}) \in E(CFI(v))$, we have \[
(v_{b'_{u_1}, b'_{u_2}, b'_{u_3}}, f((v, u_i)_{b_i})) \in E(CFI(v)) \forall i \in \{1, 2, 3\} \]  and so by Remark~\ref{obs:color restricted}  $f((v, u_i)_{b_i}) = (v, u_i)_{b_i'}$. By process of elimination, we also obtain
$f((v, u_i)_{b_i \oplus 1}) = (v, u_i)_{b_i' \oplus 1}$ for all $i \in \{1, 2, 3\}$.
In turn, by the reverse argument \[f(v_{\bar{b}_{u_1}, \bar{b}_{u_2}, \bar{b}_{u_3}}) = v_{(\bar{b}_{u_1} \oplus b_1 \oplus b'_1), (\bar{b}_{u_2} \oplus b_2 \oplus b'_2), (\bar{b}_{u_3} \oplus b_3 \oplus b'_3)}\].
\end{proof}

%
%\begin{claim}
%If $\Gamma(v) = \{u_1, u_2, u_3\}$, then for each vertex $v_{b_{u_1}, b_{u_2}, b_{u_3}} \in M(v)$,
%there is a unique automorphism $f$ of $CFR(v)$ such that $f(u_{0, 0, 0}) = v_{b_{u_1}, b_{u_2}, b_{u_3}}$.
%\end{claim}
%
%
%\begin{proof} \grant{switch proof to new notation}
%If $f(u_{0, 0, 0}) = u_{b_p, b_q, b_r}$, then since $(u_{0, 0, 0}, (u, x)_0) \in E(X(u)) \forall x \in \{p, q, r\}$, $(u_{b_p, b_q, b_r}, f((u, x)_0)) \in E(X(u))$  and so by observation 2 $f((u, x)_0) = (u, x)_{b_x}$. By process of elimination, we also obtain
%$f(u_x)_{1} = (u, x)_{b_x \oplus 1}$ for all x.
%These mappings by a similar argument force $f(u_{b'_p, b'_q, b'_r}) = u_{b_p \oplus b'_p, b_q \oplus b'_q, b_r \oplus b'_r}$.
%\end{proof}

Thus we can characterize $Aut(CFI(v))$ as being the set of even subsets of the edges of $v \in G$.\\

\paragraph{Putting together a \CFI~graph}

We define $X_f(G)$ for 3-regular G that have the property that each vertex is colored uniquely; this definition can be easily generalized to graphs
that do not have these properties.
\begin{definition}    Given a 3-regular graph $G$ such that each vertex is colored uniquely, and a function $f: E(G) \rightarrow \{0, 1\}$ we construct $X_f(G)$ as follows:
Replace each vertex $v \in G$ with $CFI(v)$, and then add edges $((v, u)_b, (u, v)_{b \oplus f((u, v))})$ for all $(v, u) \in E(G)$ and $b \in \{0, 1\}$.
\end{definition}
One way to interpret $f$ is that it specifies whether or not the edge $(u, v)$ is ``twisted" or not.
To help understand this intuition, pictures of $X_f(G)$ for a couple functions $f$ are provided in Appendix~\ref{appendix:pictures}.
\begin{definition}
We define $Y_f(G)$, an uncolored graph, as being $X_f(G)$ with all colors removed.
\end{definition}
%\aaron{This might not be terribly meaningful until the isomorphism proof for even, odd f that happens later. }

%We construct $\tilde{X}(G)$ by taking $X(G)$, and arbitrarily selecting an edge $(u, v)$, and removing edges $((u, v)_b, (v, u)_b)$ and adding edges
%$((u, v)_{b\oplus 1}, (v, u)_{b})$ for $b \in \{0, 1\}$. It is said that $\tilde{X}(G)$ is {\em corrupted with respect to} $X(G)$ on the edge $(u, v)$.\\

%Definition:
%A function $f: X(G) \rightarrow \tilde{X}(G)$ or $X(G) \rightarrow X(G)$ is said to twist the edge $(u, v) \in G$ if $f((u, v)_b) = (u, v)_{b \oplus 1}$ for all
%$b \in \{0, 1\}$. \\

%\grant{Think about if this really belongs here}
%It turns out that if $G$ is connected, $X_f(G)  \cong X_g(G)$ if and only if $f$ and $g$ have the same parity.  However, for an appropriate $G$, the Lasserre hierarchy will take many rounds to prove this fact.

%The following were proven in \cite{CFI} and are not difficult to prove:
%\begin{itemize}
%
%\item Show that if $|\{(u, v) \mid f((u, v)) = 1\}|$ is even, then $X_{\bf 0}(G)  \cong X_f(G)$.
%{\em Question: maybe remove this? It seems we essentially prove this later.}
%
%\end{itemize} 

\section{Lower-bounds for Graph Isomorphism}
In this section we show our main result.

\begin{theorem} \label{thm:main}  Let $G$ be a random 3-regular graph on n vertices and let $f$ and $g$ be functions from $E(G) \rightarrow \{0, 1\}$ of different parity.  Then 
\begin{itemize}
\item $Y_{f}(G)$ and $Y_{g}(G)$ are not isomorphic.
\end{itemize}
and with high probability:
\begin{itemize}
\item There exist vectors satisfying equations \eqref{eq:l1}-\eqref{eq:l5} on the $0.06\cdot n$ level of Lasserre for graphs $X_{f}(G)$ and $X_{g}(G)$.
\item Both $X_{f}(G)$ and $X_{g}(G)$ are $(0.01)$-expanders.
\end{itemize}
\end{theorem}

Theorem~\ref{thm:main} shows that the Lasserre hierarchy relaxations of \GI~described in \eqref{eq:las} does not provide a subexponential time algorithm for \GI~in general nor for graphs with constant expansion.

We will use the following three lemmas to prove the theorem.

%Again, fix a random 3-regular undirected graph $G$ with $n$ vertices.
%We will show that, for a given $(t,w) \in E(G)$,

\begin{lemma}\label{lemma:non-iso}   Let $G$ be a 3-regular graph, and let $f$ and $g$ be functions from $E(G) \rightarrow \{0, 1\}$ of different parity.  Then $Y_{f}(G)$ and $Y_{g}(G)$ are not isomorphic.  Moreover, if  $f$ and $g$ have the same parity, then  $X_f(G) \cong X_g(G)$.
\end{lemma}

\begin{lemma}\label{lemma:main}  Let $G$ be a 3-regular graph with with cutwidth $r$, and let $f$ and $g$ be functions from $E(G) \rightarrow \{0, 1\}$ of different parity.  Then
there exist vectors satisfying equations \eqref{eq:l1}-\eqref{eq:l5} on the $r/9$th level of Lasserre for graphs $X_{f}(G)$ and $X_{g}(G)$.
\end{lemma}

\begin{lemma}\label{lemma:CFI-expansion}  \label{theorem:CFI-expansion}
For any 3-regular graph $G$ even parity function $f: E(G) \rightarrow \{0, 1\}$, $Ex(X_f(G)) \geq \frac{1}{54} Ex(G)$.  Moreover if $f$ is odd, and $Ex(G) > 987/n$, where $n$ is the number of vertices in $G$, then $Ex(X_f(G)) \geq \frac{1}{54} Ex(G)$.
\footnote{The constants are not optimized, and there is probably some room for improvement}.
\end{lemma}

\begin{proof}[Proof Theorem~\ref{thm:main}]
The first bullet follows directly from Lemma~\ref{lemma:non-iso}.
By Theorem~\ref{thm:three-reg-exp} a random 3-regular graph has expansion $0.54$ with high probability.  Assume this is the case for some $G$.

\begin{claim}  An $n$ vertex graph $G$ with expansion $Ex(G)$ has cut-width greater than or equal to $Ex(G) n/2$.  \end{claim}
\begin{proof}  Let $\pi$ be some ordering of the vertices.  Then consider cut $E(\pi([n/2], V(G) \setminus \pi([n/2])$.  That is partition the vertices into the first half of the ordering and the second half of the ordering.  Then by the expansion of $G$ we have $E(\pi([n/2]), V(G) \setminus \pi([n/2])) = Ex(\pi([n/2]) )n/2 \geq Ex(G) n/2$.  \end{proof}

Thus by Lemma~\ref{lemma:main} there exist vectors satisfying equations \eqref{eq:l1}-\eqref{eq:l5} on the $0.54n/9 = .06$th level of Lasserre.
And also by Lemma~\ref{lemma:CFI-expansion} for large enough $n$, both $Ex(X_f(G))$ and $Ex(X_g(G))$ have expansion $0.54n/54 = .01 \cdot n$
\end{proof}

\vspace{10pt}

Our roadmap for proving these three lemmas is as follows:

\begin{itemize}

\item In this section, we first prove Lemma~\ref{lemma:non-iso}, the nonisomorphism of $Y_f(G)$ and $Y_g(G)$, if $f$ and $g$ have different parity.   This Lemma was proven in~\cite{CaiFI89} when the isomorphism must respect the colors.  However, we do not assume that $Y_f(G)$ or $Y_g(G)$ are colored graphs in this section.  Thus our results apply to the general isomorphism problem.  Elsewhere, we will assume that each vertex of $G$ has a unique color; this assumption only adds additional constraints to our vectors, so our vectors will also work in the uncolored setting.

\item Next, we set out to prove the main technical result, Lemma~\ref{lemma:main}.  We view the problem of mapping $X_{f}(G)$ to $X_{g}(G)$ as a series of linear constraints over $\F_2^n$ that we call $\varphi(G, f, g)$.  A partial isomorphism of graph vertices will map to a partial assignment of variables in such a way that a partial isomorphsim corresponds to a partial assignment that does not violate any of the constraints of $\varphi(G, f, g)$.

\item Next, we show that if $G$ has cutwidth $r$, then there is no small width resolution proof of $\varphi(G, f, g)$.

\item Next, we define vectors satisfying the Lasserre constraints of a partial isomorphism by looking at the corresponding partial assignment to
$\varphi(G, f, g)$, and then using the vectors defined in~\cite{S08} for partial assignments to formulas with no small width resolution
proof. In the following section, we then show that these vectors satisfy the Lasserre constraints of \GI.

\item Finally, we prove Lemma~\ref{lemma:CFI-expansion} which show that $X_{f}(G)$ is an expander graph as long as $G$ is, and since random graphs have constant expansion with high probability,
this will certainly be the case, showing that Lasserre cannot solve \GI even on such cases. While it seems that this should follow almost trivially from the expansion of $G$, the proof is trickier than it first appears, and as such is proved in the last section.
\end{itemize}

\subsection{Nonisomorphism of $Y_f(G)$ and $Y_g(G)$}
%\input{vectors}

%\begin{lemma}[\ref{lemma:non-iso}]  Let $G$ be a 3-regular graph let $f$ and $g$ be functions from $E(G) \rightarrow \{0, 1\}$ of different parity.  Then $Y_{f}(G)$ and $Y_{g}(G)$ are not isomorphic.
%\end{lemma}

\noindent{ \textbf{Lemma~\ref{lemma:non-iso} } (Restated)  Let $G$ be a 3-regular graph, and let $f$ and $g$ be functions from $E(G) \rightarrow \{0, 1\}$ of different parity.  Then $Y_{f}(G)$ and $Y_{g}(G)$ are not isomorphic.  Moreover, if  $f$ and $g$ have the same parity, then  $X_f(G) \cong X_g(G)$.

\begin{proof}

First, we prove that if $f$ and $g$ have the same parity, then $X_f(G) \cong X_g(G)$ (and so $Y_f(G) \cong Y_g(G)$). To do this, we show that $X_f(G)$ and $X_g(G)$ are isomorphic if
$f$ and $g$ differ in exactly 2 edges; the conclusion follows since isomorphism of two graphs is transitive. Let $(u, v)$ and $(u', v')$ be the
edges in which $f$ and $g$ differ. Consider the following isomorphism $\pi$ between $f$ and $g$: Let $P = (u_1, u_2, \ldots, u_k)$ be a simple path of adjacent edges connecting
$(u, v)$ and $(u', v')$ (so that $(u_1, u_2) = (u, v)$ or $(v, u)$ and that $(u_{k-1}, u_k) = (u', v')$ or $(v', u')$ ).
$P$ is guaranteed to exist because $G$ is connected.
Let $\pi$ be the isomorphism that for vertex $p_i$ where $2 \leq i \leq k-1$  $\pi$ ``flips'' $p_i$'s edge variables corresponding to the edges $(p_i, p_{i-1})$ and $(p_i, p_{i+1})$, that is $\pi((p_i, p_{i-1})_b) = (p_i, p_{i-1})_{b \oplus 1}$, and the middle vertices of $CFI(p_i)$ are mapped such that $\pi(p_{i~b_{v_{i-1}}, b_{v_{i+1}}, b_w}) = p_{i~b_{v_{i-1}}\oplus 1, b_{v_{i+1}} \oplus 1, b_w})$.  %analogously (see Claim~\ref{eq:unique auto}).
%Let $\pi$ be the isomorphism that for every vertex on P ``flips'' its edge vertices corresponding to edges that are on P
%(meaning that for every i, and vertex v, if $(u_i, v) \in P$ then $\pi((u_i, v)_b) = (u_i, v)'_{b \oplus 1}$.)
%\aaron{Should I make this a formal definition, since I use the word ``flips'' later?}
Observe that the ``twisted'' status of each edge is unchanged in $\pi$, except for the edges $(u, v)$ and $(u', v')$.
%Using the observation in Claim~\ref{unique auto} to determine how $\pi$ maps vertices in $CFI(u_i)$ for each i, and having $\pi$ fix other vertices,
It can be seen that $\pi$ is an isomorphism between $X_f(G)$ and $X_g(G)$.\\

We now show that if $g$ has odd parity, $Y_{g}(G) \ncong Y_{\bf 0}(G)$.
It follows that if $f$ and $g$ have different parity, then $Y_f(G) \ncong Y_g(G)$. Assume for sake of contradiction that there is some isomorphism $\pi: V(Y_{\bf 0}(G)) \rightarrow V(Y_{g}(G))$.
\begin{claim}
$\forall f, g$, if $\pi$ is an isomorphism from $Y_f(G)$ to $Y_g(G)$, $\pi$ must map middle vertices to middle vertices and edge vertices to edge vertices.
\end{claim}
\begin{proof}
It is easy to verify that the number of distinct vertices within distance 3 of a middle vertex is 19 (including the original vertex), while the number of distinct vertices within distance 3 of an edge vertex is 20 (including the original vertex).  Since isomorphisms must preserve distance between
vertices, this implies that any mapping that maps an edge vertex to a middle vertex cannot be extended to an isomorphism.
\end{proof}

Thus $\pi$ maps middle vertices to middle vertices.
Consider the induced subgraph $Y_{\bf 0}^0(G)$ of $Y_{\bf 0}(G)$, induced on the set $S_0 = \{u_{0, 0, 0} \mid u \in V(G)\} \cup \{(u, v)_0 \mid (u, v) \in E(G)\}$.
It can be verified that $Y_{\bf 0}^0(G)$ is a strict (2, 6)-stretching of $G$ (recall that this means $Y_{\bf 0}^0(G)$ can be obtained by replacing each edge in $G$ by a path of length 3).
%We show that $Y_{\bf 0}(G)\ncong Y_g(G)$ by showing that there is no partial isomorphism from $Y_{\bf 0}(G)$ to $Y_g(G)$ which map middle vertices to middle vertices and edge vertices to edge vertices.  For the sake of contradiction, assume that such a partial isomorophism existed, and
Let $S$ be the image of $Y_{\bf 0}^0(G)$.
%\begin{claim} If $g$ has odd parity, then there is no partial isomorphism from
%
%$\pi(Y_{\bf 0}^0(G))$ is not is \ncong Y_{\bf 0}^0(G)$, and thus $Y_g(G) \ncong Y_{\bf 0}(G)$.
%\end{claim}
%Let $S \subseteq Y_g(G)$ be given, and assume for sake of contradition that the induced subgraph on $S$ is isomorphic to $Y_{\bf 0}^0(G)$.
%\begin{proof}
Because $\pi$ is an isomorphism that maps middle vertices to middle vertices, $S$ has the following properties:
\begin{itemize}
\item S has $|V(G)|$ middle vertices
\item For each middle vertex in $S$, there are 3 disjoint length 3 paths to adjacent middle vertices in $S$.
\end{itemize}
Note that the second property implies that if $S$ contains a middle vertex of some $v \in V(G)$ then $S$ also contains a middle vertex
in each of $v$'s neighbors.
Thus $S$ contains exactly one middle vertex from each vertex in $G$.
For $v \in V(G)$ such that $\Gamma(v) = {u_1, u_2, u_3}$, let $v_{b_{(v, u_1)}, b_{(v, u_2)}, b_{(v, u_3)}}$ be the middle vertex in $S$.
Since for each middle vertex in $S$, there are 3 disjoint length 3 paths to other middle vertices in $S$, the set of bits $b_{(u, v)} \mid (u, v) \in E(G)$
must satisfy the following constraints: \begin{itemize}
\item $\forall v \in E(G)$, $b_{(v, u_1)} \oplus b_{(v, u_2)} \oplus b_{(v, u_3)} = 0$
\item $\forall (u, v) \in E(G)$, $b_{(u, v)} \oplus b_{(v, u)} = g((u, v))$
\end{itemize}
However, summing together all of these constraints yields the equation $0 = 1$ since g is odd, which is a contradiction. Thus there is no $S$ such that the induced subgraph of $Y_{g}(G)$ on $S$ is isomorphic to $Y_{\bf 0}^0(G)$.
%\end{proof}
%\begin{remk}
%We note that the above proof only uses the fact that middle vertices are colored differently than edge vertices. This assumption can be
%dropped by observing that middle vertices and edge vertices have different 1-WL signatures at depth 3, and so no
%bijection from $X_g(G)$ to $X_0(G)$ that maps an edge vertex to a middle vertex or vice versa will result in an isomorphism.
%Thus even if $X_g(G)$ and $X_0(G)$ for some odd g are uncolored, they are not isomorphic.
%\end{remk}
\end{proof}

\subsection{Relating Permutations to 3XOR assignments}

%\aaron{I think maybe we ought to define the variable encodings first, and then define the linear constraints, as they'll make more sense.}

Our first step is to relate partial isomorphisms on \CFI gadgets applied to 3-regular graphs to partial assignments of variables within linear equations over $\F_2^n$. This process is greatly simplified by the colors of \CFI gadgets, as the color constraints on partial isomorphisms allows us to view permutations
of vertices in terms of binary decisions. In particular, for each $(u, v) \in E(G)$, if $\sigma$ is a partial permutation, either $\forall b \in \{0, 1\}\text{ } \sigma((u, v)_b) = (u, v)'_b$ or $\forall b \in \{0, 1\} \text{ }\sigma((u, v)_b) = (u, v)'_{b \oplus 1}$.
This allows us to encode the partial permutations with partial assignments to a particular set of linear equations, which will specify constraints required for a
partial permutation to be extendable to a partial isomorphism in the context of \CFI graphs.
%We will first define the variables and their semantic meaning with respect to partial isomorphisms, and we will later
%define the constraints of the linear system, as these constriaints will make more sense given the semantics of the variables.
Our construction will create a variable $x_{(u, v)}$ for each exterior pair of vertices of an edge $(u, v) \in E(G)$. We also create variables $y_{(u, v)}$ which will encode the mapping of the internal vertices.
The semantic meaning of these constraints will then be made clear in Definition~\ref{definition:alpha}.
%A partial isomorphism $\sigma$ will relate to a partial assignment of these variables as follows:
%if for some $(u, v) \in E(G), \sigma((u, v)_0) = (u, v)'_b$, then $x_{(u, v)} = b$, and if for some $u \in V(G), \sigma

\begin{definition} \label{definition:constraints} Given $G, f, g$ we produce $\varphi(G, f, g)$, a series of linear constraints as follows:
\begin{itemize}
   \item For every vertex $v \in V(G)$ with neighbors $u_1, u_2, u_3$ create 6 Boolean variables: $x_{(v, u_i)}$ and $y_{(v, u_i)}$ for $i \in \{1, 2, 3\}$.
   \item For every vertex $v \in V(G)$ with neighbors $u_1, u_2, u_3$ create 4 constraints:
      \begin{itemize}
        \item $x_{(v, u_i)} \oplus y_{(v, u_i)} = 0$ for $i \in \{1, 2, 3\}$.
        \item $y_{(v, u_1)} \oplus y_{(v, u_2)} \oplus y_{(v, u_3)} = 0$
      \end{itemize}
   \item  For every edge $(u,v) \in E(G)$, create a constraint: $x_{(u, v)} \oplus x_{(v, u)} = f((u, v)) \oplus g((u, v))$
\end{itemize}
\end{definition}

\begin{definition} \label{definition:coordinated}
Let $\sigma$ be a color-preserving partial permutation between $X_{f}(G)$ and $X_{g}(G)$.  We say that $\sigma$ is \emph{harmonious} if it is never the case that:
$\sigma(v_{b_{u_1}, b_{u_2}, b_{u_3}}) = v_{b'_{u_1}, b'_{u_2}, b'_{u_3}}$ and $\sigma(v_{\bar{b}_{u_1}, \bar{b}_{u_2}, \bar{b}_{u_3}}) = v_{\bar{b}'_{u_1}, \bar{b}'_{u_2}, \bar{b}'_{u_3}}$ where $b_{u_i} \oplus b'_{u_i} \neq \bar{b}_{u_i} \oplus \bar{b}'_{u_i}$ for some $i \in \{1, 2, 3\}$
\end{definition}

\begin{definition} \label{definition:alpha}
Let $\sigma$ be a harmonious partial permutation between $X_{f}(G)$ and $X_{g}(G)$.  We define $\alpha_\sigma$, a partial assignment to $\varphi(G, f, g)$,
as follows:
\begin{itemize}
   \item  If $\sigma$ maps $(v, u)_b$ to $(v, u)_{b'}$  then $x_{(u, v)} = b \oplus b'$.
   \item  If $\sigma$ maps $v_{b_{u_1}, b_{u_2}, b_{u_3}}$ to $v_{b'_{u_1}, b'_{u_2}, b'_{u_3}}$, then $y_{(v, u_k)} = b_{u_k} \oplus b'_{u_k}$ for each $k \in \{1, 2, 3\}$.
 \end{itemize}
Note that $\alpha_{\sigma}$ is well-defined because $\sigma$ is harmonious.
\end{definition}

%\begin{definition}  \label{definition:h}
%For $\sigma \in P(G)$, define the function $h_\sigma : \{0, 1\}^{|E(G)|} \rightarrow \{0, 1\}$, indexing the input bits with edges of $G$, such that for
% harmonious partial permutation $\sigma$, $h_\sigma({\bf w}) =
%1$ iff for every $x_{(u, v)} = b_i$ specified by $\alpha_{\sigma}$, $w_{(u, v)} = b_i$.
%If $\sigma$ is not harmonious or $\sigma = \bot$, then $h_\sigma = {\bf 0}$.
%\end{definition}

%We then create constraints to make sure that $|\{u_i: y_{(v, u_i)}=1\}|$ is even, and that the mapping of the internal vertices is consistant with the mapping of the external vertices.
In light of Definition~\ref{definition:alpha}, the constraints in Definition~\ref{definition:constraints} can be viewed as the following constraints on partial permutations to be extendable to a partial isomorphism from $X_f(G)$ to $X_g(G)$:
\begin{itemize}
\item For any  $v \in V(G)$ with $\Gamma(v) = \{u_1, u_2, u_3\}$; for any $i \in \{1, 2, 3\}$ if $(v, u_i)_b \in dom(\sigma)$ and $v_{b_{u_1}, b_{u_2}, b_{u_3}} \in dom(\sigma)$ for some values of $b, b_{u_1}, b_{u_2}, b_{u_3}$,
then $\sigma$ must preserve the edge relation between these two vertices.
%two vertices must be preserved. \aaron{I really haven't explained $\sigma$ ``flipping'' things.}
\item For any $v \in V(G)$, if $\Gamma(v) = \{u_1, u_2, u_3\}$, and $v_{b_{u_1}, b_{u_2}, b_{u_3}} \in dom(\sigma)$ for some values of $b_{u_1}, b_{u_2}, b_{u_3}$, $\sigma$ must change an even number of $b_{u_1}, b_{u_2}, b_{u_3}$, since otherwise the vertex $\sigma$ is
mapping $v_{b_{u_1}, b_{u_2}, b_{u_3}}$ to does not exist.
\item For any $(u, v) \in E(G)$, if $(u, v)_b \in dom(\sigma)$ and $(v, u)_{\bar{b}} \in dom(\sigma)$ for some values of $b, \bar{b}$,
then $\sigma$ must preserve the edge relation between $(u, v)_b$ and $(v, u)_{\bar{b}}$.
\end{itemize}

\subsection{$\varphi(G, f, g)$ requires high-width refutation proofs}

Let $G$ be a graph with cutwidth $r+1$, and let $f, g$ be functions from $E(G) \rightarrow \{0, 1\}$.

Now we build some tools up to reason about the partial assignments of $\varphi(G, f, g)$. %This analysis is functionally equivalent to width-$w$ resolution on the constraints of $\varphi(G, f, g)$ as seen in \cite{S08}.
The following defines a ``step'' in the resolution process.

\begin{definition}
Given sets $S, T \subseteq E(G)$, we say that $S \vdash_+ T$ if there exists a constraint $\phi \in \varphi(G, f, g)$ such that  \[\bigg(\bigoplus_{(u, v) \in S} x_{(u, v)} = 0\bigg) \oplus \phi = \bigg(\bigoplus_{(u, v) \in T} x_{(u, v)} = 0\bigg)\]
We say that $S \vdash_- T$ if there exists a constraint $\phi \in \varphi(G, f, g)$ such that \[\bigg(\bigoplus_{(u, v) \in S} x_{(u, v)} = 0\bigg) \oplus \phi = \bigg(\bigoplus_{(u, v) \in T} x_{(u, v)} = 1\bigg)\]
\end{definition}

With the notion of a ``step'' in hand, we now define the notion of a width-$w$ ``proof'' as follows:
%Note that the presentation of width-(c/3) resolution as
%presented here is functionally identical to that of \cite{Grant} but we use the features of $\phi_G$ to simplify notation.\\
%Let $\mathbb{L} := \{S \subseteq E(G) \mid |S| \leq cn/6\}$.\\
%Let 2$\mathbb{L} := \{S \subseteq E(G) \mid |S| \leq cn/3\}$.\\
%Given $S, T \in 2\mathbb{L}$, we say that $S \vdash_+ T$ if there exists a vertex w such that \[\bigg(\bigoplus_{(u, v) \in S} x_{(u, v)} = 0\bigg) \oplus \phi_w = \bigg(\bigoplus_{(u, v) \in T} x_{(u, v)} = 0\bigg)\]
%We say that $S \vdash_- T$ if there exists a vertex w such that \[\bigg(\bigoplus_{(u, v) \in S} x_{(u, v)} = 0\bigg) \oplus \phi_w = \bigg(\bigoplus_{(u, v) \in T} x_{(u, v)} = 1\bigg)\]
%(Note that this w must be the vertex next to the ``corrupted edge'' chosen.)
%\aaron{It feels like this definition should also include f and g. But sticking more subscripts and superscripts on a relation symbol feels weird.}
\begin{definition}
Given sets $S, T \subseteq E(G)$ satisfying $|S| \leq w$ and $|T| \leq w$, we say that $S \sim_{w}^+ T$ if there exists a finite sequence $\{S_i \mid i \leq t\}$ such that:
\begin{itemize}
\item $|S_i| \leq w$ for all i,
\item $S_0 = S$, $S_t = T$,
\item For all $i < t$, either $S_i \vdash_+ S_{i+1}$ or $S_i \vdash_- S_{i+1}$.
\item The number of i such that $S_i \vdash_- S_{i+1}$ is even.
\end{itemize}
We say that $S \sim_{w}^{-} T$ similarly, replacing ``even'' with ``odd'' in the last item.
Furthermore, we say that $S$ implies $T$ via \emph{width-$w$ resolution} denoted $S \sim_w T$ if either $S \sim_{w}^+ T$ or  $S \sim_{w}^{-} T$.
%We say that $\phi$ is used $j$ times in the proof of $T$ given $S$ if for $j$ different i,
%\[ \[\bigg(\bigoplus_{(u, v) \in S} x_{(u, v)} = 0\bigg) \oplus \phi = \bigg(\bigoplus_{(u, v) \in T} x_{(u, v)} = b\bigg)\].
\end{definition}

Intuitively, $S \sim_{w}^+ T$ means that given the parity of S we can prove using width-$w$ resolution on $\varphi(G, f, g)$ that the parity of $T$ is equal to the parity of $S$, and $S \sim_{w}^{-} T$ means that given the parity of $S$ we can prove using width-$w$ resolution on $\phi_G$ that the parity of $T$ is opposite to the parity of $S$.

\begin{proposition}\label{proposition:widthw}  Let $G$ have cutwidth $r+1$, then 
$\varphi(G, f, g)$ cannot be refuted by width-$r$ resolution.
\end{proposition}

\begin{proof}
We show that $\emptyset \not\sim_{r}^{-} \emptyset$. Suppose that $\emptyset \sim_{w}^{-} \emptyset$ .
Let $\{S_0, \ldots, S_k\}$ be a sequence that is a witness to this fact. We show that $w > r$.\\

\begin{claim}Each clause is used in an odd number of steps of $\{S_0, \ldots, S_k\}$.\end{claim}

\begin{proof}
Note that each $v \in G'$ must appear in an even number of $\phi$
that are used in the proof, since $S_k = \emptyset$.  However, all variables of $\varphi(G, f, g)$ appear in exactly 2 constraints in $\varphi(G, f, g)$.
Thus if a constraint $\phi$ is used an odd number of times in the proof, each constraint sharing a variable with $\phi$ must also be used an odd number of times. Since $G$ is connected, this implies that if any constraint is used an odd number of times, then all constraints are.  However, some constraint must be used an odd number of times because by definition of $\{S_0, \ldots, S_k\}$ there are an odd number of $i$ such that $S_i \vdash_- S_{i+1}$. This
implies that there must exist a constraint containing the constant 1 appearing an odd number of times.
And thus, at $S_k$ all constraints have been used an odd number of times.
\end{proof}
%Since no constraint has been used an odd number of times prior to $S_0$ and all of them have been used an odd number of times in $S_k$.

We consider when the  middle vertex constraints $y_{(v, u_1)} \oplus y_{(v, u_2)} \oplus y_{(v, u_3)} = 0$ are used.  Let the set $\omega(S_i) \subseteq V(G)$ be defined as containing the vertices whose corresponding middle vertex constraints have been used an odd number of times in resolving from $S_0$ to $S_i$.  Then $|S_i| \geq E(\omega(S_i), V(G) \setminus \omega(S_i))$ because there is at least one variable in $S_i$ corresponding to each edge in $E(\omega(S_i), V(G) \setminus \omega(S_i))$.

By Theorem~\ref{theorem:CW(G)=W(G)} we know that if $G$ has cut-width $r+1$, then there exists a monotone set $\Omega$ of $\mathcal{P}(V(G))$ such that
$r +1=\min_{(S_1, S_2) \in \partial\Omega} \max_{i \in \{1, 2\}} E(S_i, V(G) \setminus S_i)$.  Recall  $(S_1, S_2) \in \partial\Omega$ if $S_1 \in \Omega$, $S_2 \not\in \Omega$ and $S_1 = S_2 \cup \{i\}$ for some element $i$.

Then $\omega(S_0) = \emptyset$ and  $\omega(S_k) = V(G)$.  Thus it must be the case, that $(\omega(S_i), \omega(S_{i+1})) \in \partial\Omega$ for some $i$.  However, at this point, $\max_{j \in \{i, i+1\}} E(S_j, V(G) \setminus S_j) \geq r+1$, and thus $S_j$ contains at least $r+1$ variables.
\end{proof}

\begin{corollary} \label{corollary:stable}
If G has cutwidth r+1, then for any $S, T$ with at most $r/3$ variables, it is not possible for $S \sim_{2r/3}^+ T$ and $S \sim_{2r/3}^- T$ to be
simultaneously true.
\begin{proof}
Suppose that $S \sim_{2r/3}^+ T$ and $S \sim_{2r/3}^- T$ for some $S, T$ with at most $r/3$ variables.
The same proof that shows $S \sim_{2r/3}^+ T$ shows that $\emptyset \sim_{r}^+ S \oplus T$,
and the reverse of the proof that shows $S \sim_{2r/3}^- T$ shows that $S \oplus T \sim_{r}^- \emptyset$.
Putting these two proofs together, we obtain that $\emptyset \sim_{r}^- \emptyset$, contradicting Proposition~\ref{proposition:widthw}.
\end{proof}
\end{corollary}

\begin{remk}
Note that $\sim_w$ is an equivalence relation for all w.
\end{remk}

\subsection{Vectors satisfying equations \eqref{eq:l1}-\eqref{eq:l5}}

%\aaron{Should this definition (below) be moved to the preliminaries, and made more general?}
\begin{definition}
Let $\mathbb{L}_{r/3}$ be the set of subsets of at most $r/3$ variables of $\varphi(G, f, g)$
\end{definition}
\begin{definition}
Define $E_{r/3} = \mathbb{L}_{r/3} / \sim_{2r/3}$ as the set of equivalence classes over $\sim_{2r/3}$;
for each equivalence class $e \in E_{r/3}$, we arbitrarily choose an exemplar $S_0 \in [S_0]_{r/3}$.
Given an equivalence class $[S]$, we will denote its exemplar as $[S]^0$.\\
We will drop the subscript if its value is clear from context.
\end{definition}

For sake of notational convenience, $\emptyset \in [\emptyset]^0$; note that $[\emptyset]_{r/3}$ is the class of elements of
$\mathbb{L}_{r/3}$ which we know the parity of using a width-$2r/3$ resolution proof without any other assumptions.

\begin{definition}
We define a function $\gamma: \mathbb{L}_{r/3} \rightarrow \{-1, +1\}$ by \[\gamma(S) =
\begin{cases}
1 & \text { if }  [S]^0 \sim_{2r/3}^+ S\\
-1 & \text { if } [S]^0 \sim_{2r/3}^- S\\
 \end{cases}\]
%The choice of $[S]^0$ will affect the sign of $\gamma(S)$ for each S, but it has the same effect for each S.
%For our purposes, this does not matter.
Note that this definition of $\gamma$ is well-defined by Corollary \ref{corollary:stable}.
\end{definition}

\begin{definition}  \label{definition:h}
For $\sigma \in P(G)$, define the function $h_\sigma : \{0, 1\}^{|E(G)|} \rightarrow \{0, 1\}$, indexing the input bits with edges of $G$, such that for
 harmonious partial permutation $\sigma$, $h_\sigma({\bf w}) =
1$ if for all $(u, v) \in E(G)$: $x_{(u, v)}, x_{(v, u)}, y_{(u, v)}, y_{(v, u)}$ are either undefined by $\alpha_{\sigma}$ or equal to $w_{(u, v)}$.  Otherwise $h_\sigma({\bf w}) =
0$.  If $\sigma$ is not harmonious or $\sigma = \bot$, then $h_\sigma = {\bf 0}$.
\end{definition}

With these definitions in hand, 
%of width-$r/3$ resolution in hand for $\varphi(G, f, g)$, 
we are ready to define the proposed vectors.  
\begin{definition} \label{definition:vectors}
Let $\sigma$ be a partial function mapping $X_f(G)$ to $X_g(G)$ such that $|dom(\sigma)| \leq r/9$. If $\sigma$ is harmonious
and color-coordinated, then
\[v_\sigma = \sum_{S \in \mathbb{L}_{r/3}} \widehat{h_\sigma}(S)\gamma(S)e_{[S]}\]
%\aaron{It could be confusing here that $\alpha(f)$ is itself a function - any notation that may help here?}
where $e_{[S]}$ is a vector with $|E_{r/3}|$ coordinates; the coordinate indexed by $[S]$ is 1 and all remaining coordinates are 0.  Otherwise, $v_\sigma = {\bf 0}$.
\end{definition}

\section{Proof that these vectors satisfy Lasserre Constraints}
\noindent \textbf{Lemma~\ref{lemma:main}}(Restated) Let $G$ be a 3-regular graph with with cutwidth $r$, and let $f$ and $g$ be functions from $E(G) \rightarrow \{0, 1\}$ of different parity.  Then
there exist vectors satisfying equations \eqref{eq:l1}-\eqref{eq:l5} on the $r/9$th level of Lasserre for graphs $X_{f}(G)$ and $X_{g}(G)$.

%
%\begin{theorem}
%The vectors defined in \ref{definition:vectors} satisfy the constraints in \cref{eq:l1}-\cref{eq:l5}
%\end{theorem}

\begin{lemma}
\eqref{eq:l1}: $||v_{\emptyset}|| = 1$
\end{lemma}
\begin{proof}
In this case, the function $h_\emptyset = {\bf 1}$ and so $\widehat{h}_\emptyset(\chi_I) = 1$ if $I = \emptyset$ and 0 otherwise.
Therefore, $v_{\emptyset} = e_{[\emptyset]}$, and thus $||v_{\emptyset}|| = 1$.
\end{proof}

\begin{lemma}
\eqref{eq:l2}: $\forall (i, j)\text{  }\sum_{i', j'} \langle v_{i\rightarrow i'}, v_{j \rightarrow j'} \rangle B_{i'j'} = A_{ij}$
\end{lemma}
\begin{proof}

First, we explicitly compute the vectors for $\sigma$ that map just one vertex using Definitions~\ref{definition:alpha} and \ref{definition:vectors}.
These computations are straightforward but cumbersome, and are presented in Appendix~\ref{appendix:vectorcomputation}. 
%The most interesting
%observation used here is that for any vertex, any subset of its neighbors is in the same equivalence class (under $\sim_w$).
\begin{itemize}
\item For $(u, v) \in E(G)$,  \[v_{(u, v)_b \rightarrow (u, v)_c'} =
\frac{1}{2}e_{[\emptyset]} +
 \frac{1}{2}(-1)^{b \oplus c}\gamma((u, v))e_{[(u, v)]}\]

\item For $u \in V(G)$ with $\Gamma(u) = \{u_1, u_2, u_3\}$, assuming that $\gamma({u, u_i}) = 1$ for all i, we have
\[v_{u_{b_{u_1}, b_{u_2} b_{u_3}} \rightarrow u'_{\bar{b}_{u_1}, \bar{b}_{u_2}, \bar{b}_{u_3}}} =
\frac{1}{4} e_{[\emptyset]} + (-1)^{b_{u_1} \oplus \bar{b}_{u_1}} \frac{1}{4} e_{[(u, u_1)]} + (-1)^{b_{u_2} \oplus \bar{b}_{u_2}}
\frac{1}{4} e_{[(u, u_2)]} + (-1)^{b_{u_3 \oplus \bar{b}_{u_3}}}\frac{1}{4} e_{[(u, u_3)]}\]
\end{itemize}

Now, we partition the $(i, j)$ pairs into the following 3 cases, and show that \eqref{eq:l2} holds in all of them:
\begin{itemize}
\item Case 1: $(i, j)$ such that that $E(C(i), C(j)) = 0$
\item Case 2: $(i, j)$ such that $\exists w$ such that $i \in M(w)$ and $j \in E(w)$ or vice versa
\item Case 3: $(i, j)$ such that $\exists (u, w) \in E(G), b_i, b_j$ such that $i = (u, w)_{b_i}$ and $j = (w, u)_{b_j}$
\end{itemize}
%It can be readily verified that this partitions the set of pairs of vertices of $X_f(G)$.\\

Case 1: If $i$ and $j$ are such $|E(C(i), C(j))| = 0$, then $A_{ij} = 0$, and
for all $B_{i'j'} \neq 0$,
$\langle v_{i\rightarrow i'}, v_{j \rightarrow j'} \rangle = 0$, because otherwise $|E(C(i'), C(j'))| >0$ and so either $c(i) \neq c(i')$ or $c(j) \neq c(j')$.
Therefore either $i \rightarrow i'$ or $j \rightarrow j'$ fails be to color preserving  
%$\alpha_{i\rightarrow i'}$ or  $\alpha_{j \rightarrow j'}$ 
and therefore either $v_{i \rightarrow i'}$ or $v_{j \rightarrow j'}$ is ${\bf 0}$.
Thus $\langle v_{i \rightarrow i'}, v_{j \rightarrow j'} \rangle = 0$ either way, and so all terms in $\sum_{i', j'} \langle v_{i\rightarrow i'}, v_{j \rightarrow j'} \rangle B_{i'j'} = A_{ij}$ are 0, satisfying the equality in this case. \\

Case 2: Suppose there exists $w$ such that $i \in M(w)$ and $j \in E(w)$. In particular, let $\Gamma(w) = \{u_1, u_2, u_3\}$,
and suppose without loss of generality that $i = w_{b_{u_1}, b_{u_2}, b_{u_3}}$ and $j = (w, u_1)_{c}$.
Since $v_{k \rightarrow k'} = 0$ if $c(k) \neq c(k')$, we restrict our summation to be over $i'$ and $j'$ that share a color with $i$ and $j$, respectively.
Thus we only consider $(i', j')$ pairs such that $i' = w'_{b'_{u_1}, b'_{u_2}, b'_{u_3}}$ and $j' = (w, u_1)'_{c'}$, with
$u_1 \oplus u_2 \oplus u_3 = 0$, and so
\[\langle v_{i\rightarrow i'}, v_{j \rightarrow j'} \rangle =
\frac{1}{8} + (-1)^{c \oplus c' \oplus b_{u_1} \oplus b'_{u_1}}
\frac{1}{8}\]
Noting that $c \oplus c' \oplus b_{u_1} \oplus b'_{u_1} = 0 \iff (c \oplus b_{u_1}) = (c' \oplus b'_{u_1})$,
\[ \langle v_{i\rightarrow i'}, v_{j \rightarrow j'} \rangle =
\begin{cases}
\frac{1}{4} & \text{ if } (c \oplus b_{u_1}) = (c' \oplus b'_{u_1})\\
0 & \text{ otherwise }
\end{cases}
\]
Recall that by definition of $X_f(G)$ and $X_g(G)$ that $A_{ij} = 1 - (c \oplus b_{u_1})$ and $B_{i'j'} = 1 - (c' \oplus b'_{u_1})$.
Thus \[ \langle v_{i\rightarrow i'}, v_{j \rightarrow j'} \rangle =
\begin{cases}
\frac{1}{4} & \text{ if } A_{ij} = B_{i'j'}\\
0 & \text{ otherwise }
\end{cases}
\]
Thus if $A_{ij} = 0$, $\sum_{i'j'} \langle v_{i\rightarrow i'}, v_{j \rightarrow j'} \rangle B_{i'j'} = 0$, as each of its terms must be 0.
If $A_{ij} = 1$, then $B_{i'j'} = 1$ for all nonzero $\langle v_{i\rightarrow i'}, v_{j \rightarrow j'} \rangle$, and so
 \[ \sum_{i', j'} \langle v_{i\rightarrow i'}, v_{j \rightarrow j'} \rangle B_{i'j'} = \sum_{i', j'} \langle v_{i\rightarrow i'}, v_{j \rightarrow j'} \rangle \]
\[= \frac{1}{4}\cdot|\{c', u'_1, u'_2, u'_3 \mid ((c \oplus c') = (u_1 \oplus u'_1)) \text{ and } (u_1 \oplus u_2 \oplus u_3 = 0)\}|
\]
There are two solutions to the constraint $((c \oplus c') = (u_1 \oplus u'_1))$, and with fixed $u_1$, there are 2 solutions to
$u_1 \oplus u_2 \oplus u_3 = 0$, making for a total of 4 solutions.
Thus  \[ \sum_{i', j'} \langle v_{i\rightarrow i'}, v_{j \rightarrow j'} \rangle B_{i'j'} = \frac{1}{4} \cdot 4 = 1
\]
%With the same i and j, we compute $\sum_{i', j'} \langle v_{i\rightarrow i'}, v_{j \rightarrow j'} \rangle B_{i'j'} $.
%We can restrict our summation to i' and j' of the form $i' = (u, v)_{c}'$ and $j' =  u_{c_v, c_p, c_q}'$.
%Note that if $b \oplus b_v = 1$ and thus $A_{ij} = 0$, in order for $\langle v_{i\rightarrow i'}, v_{j \rightarrow j'} \rangle > 0, c \oplus c_v = 1$ as well.
%Thus $B_{i'j'} = 0$ in that case, and so $\sum_{i', j'} \langle v_{i\rightarrow i'}, v_{j \rightarrow j'} \rangle B_{i'j'}  = 0 = A_{ij}$ in that case.\\
%If $b \oplus b_v = 0$, then there are 4 different $i',j'$ pairs such that $\langle v_{i\rightarrow i'}, v_{j \rightarrow j'} \rangle = \frac{1}{4}$, and each
%has $B_{i'j'} = 1$.
%there are two ways to make $c \oplus c_v = 0$, and for each of those ways 2 ways to make $c_p \oplus c_q \oplus c_v = 0[1 \text{ if twisted vtx.}]$
%Therefore $\sum_{i', j'} \langle v_{i\rightarrow i'}, v_{j \rightarrow j'} \rangle B_{i'j'} = 1$.\\
Case 3: Suppose there exists an edge $(u, w)$ such that $i = (u, w)_{b_i}$ and $j = (w, u)_{b_j}$.
Since $v_{k \rightarrow k'} = 0$ if $c(k) \neq c(k')$, we restrict our summation to be over $i'$ and $j'$ sharing a color with $i$ and $j$, respectively.
Thus we only consider $(i', j')$ pairs such that $i' = (u, w)'_{b'_i}$ and $j' = (w, u)'_{b'_j}$.
Without loss of generality, we assume that $\gamma(x_{(u, w)}) = 1$.
Note that with this assumption, $\gamma(x_{(w, u)}) = -1^{f((u, v)) \oplus g((u, v))}$.
The value of $\langle v_{i\rightarrow i'}, v_{j \rightarrow j'} \rangle$ therefore is $\frac{1}{4} +
(-1)^{b_i \oplus b'_i \oplus b_j \oplus b'_j \oplus f((u, v)) \oplus g((u, v))}\frac{1}{4}$.
Thus \[\langle v_{i\rightarrow i'}, v_{j \rightarrow j'} \rangle =
\begin{cases}
\frac{1}{2} & \text{ if } b_i \oplus b'_i \oplus b_j \oplus b'_j \oplus f((u, v)) \oplus g((u, v))= 0\\
0 & \text{ otherwise. }\\
\end{cases}
\]
Furthermore, we note that $A_{ij} = 1$ if and only if $b_i \oplus b_j \oplus f((u, v)) = 0$,
and $B_{i'j'} = 1$ if and only if $b'_i \oplus b'_j \oplus g((u, v)) = 0$.
Since $b_i \oplus b'_i \oplus b_j \oplus b'_j \oplus f((u, v)) \oplus g((u, v))= 0 \iff (b_i \oplus b_j \oplus f((u, v))) = (b'_i \oplus b'_j \oplus g((u, v))$,
we have that  \[\langle v_{i\rightarrow i'}, v_{j \rightarrow j'} \rangle =
\begin{cases}
\frac{1}{2} & \text{ if } A_{ij} = B_{i'j'}\\
0 & \text{ otherwise. }\\
\end{cases}
\]
Thus if $A_{ij} = 0$, then $\sum_{i', j'} \langle v_{i\rightarrow i'}, v_{j \rightarrow j'} \rangle B_{i'j'}$ = 0 since all terms are 0.
If $A_{ij} = 1$, then $B_{i'j'} = 1$ for all nonzero $\langle v_{i\rightarrow i'}, v_{j \rightarrow j'} \rangle$, and so
 \[ \sum_{i', j'} \langle v_{i\rightarrow i'}, v_{j \rightarrow j'} \rangle B_{i'j'} = \sum_{i', j'} \langle v_{i\rightarrow i'}, v_{j \rightarrow j'} \rangle \]
\[= \frac{1}{2} \cdot |\{b'_i, b'_j \mid b_i \oplus b'_i \oplus b_j \oplus b'_j \oplus f((u, v)) \oplus g((u, v)) = 0)\}|
\]
Fixing the other 4 variables, there are 2 choices of $(b'_i, b'_j)$ such that $b_i \oplus b'_i \oplus b_j \oplus b'_j \oplus f((u, v)) \oplus g((u, v)) = 0$,
and so
\[ \sum_{i', j'} \langle v_{i\rightarrow i'}, v_{j \rightarrow j'} \rangle B_{i'j'} = \frac{1}{2} \dot 2 = 1\]

%If $b \oplus p = 0$, then there are 2 different i', j' such that $\langle v_{i\rightarrow i'}, v_{j \rightarrow j'} \rangle = \frac{1}{2}$
%as there are 2 ways to make $b \oplus c \oplus p \oplus q = 0$, both of which have $B_{i'j'} = 1$.
%Otherwise, $c \oplus q = 1$ for all i', j' such that $B_{i'j'} = 1$, so $\langle v_{i\rightarrow i'}, v_{j \rightarrow j'} \rangle B_{i'j'} = 0$ for all i', j'.

\end{proof}

\begin{lemma}
{\eqref{eq:l3} $\forall \sigma_1, \sigma_2 \text{ s.t. } \sigma_1 \wedge \sigma_2 = \sigma_1' \wedge \sigma'_2,
\langle v_{\sigma_1}, v_{\sigma_2} \rangle = \langle v_{\sigma_1'}, v_{\sigma_2'} \rangle$ }
\end{lemma}

\begin{proof}
\begin{claim}
For all $\sigma_1$, $\sigma_2$,
$h_{\sigma_1} \cdot h_{\sigma_2} = h_{\sigma_1 \wedge \sigma_2}$.
\end{claim}
\begin{proof}
We split the proof into 2 main cases.
If $\sigma_1$ and $\sigma_2$ are consistent and ${\sigma_1 \wedge \sigma_2}$ is color-coordinated and harmonious, $\sigma_1$ and $\sigma_2$ are both harmonious, and so by construction of h and definition of $\wedge$, we have $h_{\sigma_1} \cdot h_{\sigma_2} = h_{\sigma_1 \wedge \sigma_2}$.
If these conditions are not satisfied, we show that $h_{\sigma_1} \cdot h_{\sigma_2} = 0$, as $h_{\sigma_1 \wedge \sigma_2} = h_\bot = 0$ in these cases.
\begin{itemize}
\item If $\sigma_1 \wedge \sigma_2$ is not color-coordinated,
then one of $\sigma_1$ or $\sigma_2$ is not color-coordinated, and so $h_{\sigma_1} = {\bf 0}$ or $h_{\sigma_2} = {\bf 0}$, and thus
$h_{\sigma_1} \cdot h_{\sigma_2} = {\bf 0}$.
\item If $\sigma_1$ and $\sigma_2$ are inconsistent,
then $h_{\sigma_1} \cdot h_{\sigma_2} = {\bf 0}$ by definition of $h$, as they encode partial permutations that are mutually exclusive.
\item If $\sigma_1 \wedge \sigma_2$ is consistent but not harmonious, then if $\sigma_1$ or $\sigma_2$ are not harmonious,
then $h_{\sigma_1} = {\bf 0}$ or $h_{\sigma_2} = {\bf 0}$. Otherwise, let u be a vertex that is a witness
to the fact that $\sigma_1$ and $\sigma_2$ are not harmonious, and for $i \in \{1, 2\}$, construct $\tilde{\sigma_i}$
to be the same as $\sigma_i$, except that $\tilde{\sigma_i}$ is defined for all vertices in $M(u)$ for each. By construction,
$h_{\tilde{\sigma_i}} = h_{\sigma_i}$, but $\tilde{\sigma_1}$ and $\tilde{\sigma_2}$ are not consistent, and so
$h_{\sigma_1} \cdot h_{\sigma_2} = h_{\tilde{\sigma_1}} \cdot h_{\tilde{\sigma_2}} = {\bf 0}$.
\end{itemize}
From this, we know that
$h_{\sigma_1} \cdot h_{\sigma_2} = h_{\sigma_1'} \cdot h_{\sigma_2'}$ if $\sigma_1 \wedge \sigma_2 = \sigma_1' \wedge \sigma_2'$.
%\begin{claim}
%Suppose that $\sigma_1$ and $\sigma_2$ are not consistent or that $\sigma_1 \wedge \sigma_2$ is not harmonious (and thus $v_{\sigma_1 \wedge \sigma_2} = 0)$. In this case,
%$\langle v_{\sigma_1} , v_{\sigma_2} \rangle = 0$.
%\end{claim}
%\begin{proof}
%We assume that $\sigma_1$ and $\sigma_2$ are both color-coordinated, and thus $\sigma_1 \wedge \sigma_2$ is;
%if $\sigma_1$ or $\sigma_2$ is not color-coordinated, $v_{\sigma_1}$ or $v_{\sigma_2} = 0$, implying the conclusion.
%\end{proof}
%Note that if $\sigma_1$
%This implies that $\eqref{eq:l3}$ holds if any of $\sigma_1, \sigma'_1, \sigma_2, \sigma'_2$ are not color-coordinated.
%So we assume that $\sigma_1, \sigma'_1, \sigma_2, \sigma'_2$ are all color-coordinated. We note that  $h_{\sigma_1 \wedge \sigma_2}$
%is 1 iff the constraints of both $\alpha_{\sigma_1}$ and $\alpha_{\sigma_2}$ are satisfied, and so if $\sigma_1 \wedge \sigma_2 = \sigma'_1 \wedge \sigma'_2$
%then $h_{\sigma_1 \wedge \sigma_2} = h_{\sigma_1 \wedge \sigma_2}$.
The rest of the proof follows that of $\cite{S08}$, following from the fact that the fourier coefficients of $h_{\sigma_1 \wedge \sigma_2}$ and
$h_{\sigma'_1 \wedge \sigma'_2}$ are the same, and thus we can write $\langle v_{\sigma_1}, v_{\sigma_2} \rangle$ in terms of only $\widehat{h_{\sigma_1 \wedge \sigma_2}}$, the fourier coefficients of $h_{\sigma_1 \wedge \sigma_2}$.

It will follow that if $\sigma_1 \wedge \sigma_2 \equiv \sigma'_1 \wedge \sigma'_2$ then  $\langle v_{\sigma_1}, v_{\sigma_2} \rangle  =  \langle v_{\sigma'_1}, v_{\sigma'_2} \rangle$  because the fourier coefficients of $h_{\sigma_1 \wedge \sigma_2}$ and $h_{\sigma_1 \wedge \sigma_2}$.  For $[S] \in \mathbb{E}_{r/3}$, let $\widehat{h_{\sigma_1}}([\chi_S]) =
\sum_{S' \in [S]}  \widehat{h_{\sigma_1}}(S')\gamma(S')$ .  Then
\begin{eqnarray*}
  \langle v_{\sigma_1}, v_{\sigma_2} \rangle & = & \sum_{[S] \in \mathbb{E}_{r/3}} \langle \widehat{h_{\sigma_1}}([S]), \widehat{h_{\sigma_2}}([S]) \rangle \\
                            & = & \sum_{S  \in \mathbb{L}_{\frac{r}{3}}}  \widehat{h_{\sigma_1}}(S) \gamma(S)  \sum_{T \in [S]} \widehat{h_{\sigma_2}}(T)\gamma(S) \\
                           & = & \sum_{S  \in \mathbb{L}_{\frac{r}{3}}}  \widehat{h_{\sigma_1}}(S) \gamma(S) \sum_{U \in \mathcal{[\emptyset]}} \widehat{h_{\sigma_2}}(SU) \gamma(SU)  \\
                           & = & \sum_{U \in [\emptyset]} \gamma(U) \sum_{S  \in \mathbb{L}_{\frac{r}{3}}} \widehat{h_{\sigma_1}}(S)\widehat{h_{\sigma_2}}(S U)\\
                           & = &  \sum_{U \in [\emptyset]} \gamma(U) \widehat{h_{\sigma_1} \cdot {h_{\sigma_2}}}(U) \\
                           & = &  \sum_{U \in [\emptyset]} \gamma(U) \widehat{h_{\sigma_1 \wedge \sigma_2}}(U) \end{eqnarray*}
The second line follows from expanding the summands.  The third line follows from the fact that $[S] \subseteq S \cdot [\emptyset]$ and because $h_{\sigma_2}$ is a $\frac{r}{3}$-junta (since $\sigma$ has $\frac{r}{9}$ mappings, each of which induce a dependence on at most 3 bits of $h_{\sigma_2}$), $\widehat{h_{\sigma_2}}(S \Delta U) = 0$ if the size of $S \Delta U$ is greater than $\frac{r}{3}$.  The fourth line follows because $\gamma(S)\gamma(SU) = \gamma(U)$, and the fifth line from the fact that $ \widehat{h_{\sigma_1}} \cdot \widehat{h_{\sigma_2}}(U) = \sum_{S  \in \mathcal{L}^{\frac{r}{3}}} \widehat{h_{\sigma_1}}(S)\widehat{h_{\sigma_2}}(S U)$ because the full fourier expansions of $\widehat{h_{\sigma_1}}$ and $\widehat{h_{\sigma_2}}$ are captured by the characters of $\mathbb{L}_{\frac{r}{3}}$.

%We define \[
%\lbrace v_{\sigma_1}, v_{\pi_2} \rbrace =
%, if $\pi_1 \wedge \pi_2 = \pi_1' \wedge \pi_2'$, then
%\[v_{\pi_1 \wedge \pi_2} =
%\sum_{S \in \mathbb{L}} \widehat{\alpha_{\pi_1 \wedge \pi_2}(S)\gamma(S)e_{[S]}\]
%If $\pi_1 \wedge \pi_2$ is color-coordinated, then
%Schoenebeck in \cite{Grant} shows that this is equal to
%\[\sum_{I \subseteq E(G), |I| \leq cn/6} \widehat{h_{f' \circ g'}}(\chi_I)\pi(\chi_I)e_{[I]}\] which is simply $v_{f' \circ g'}$.

\end{proof}
\end{proof}
\begin{lemma}
\eqref{eq:l4}: $\forall i \in V(X_f(G)), \text{  } v_\sigma = \sum_{i' \in V(X_g(G))} v_{\sigma \wedge (i \rightarrow i')}$ and
\eqref{eq:l5}: $\forall i' \in V(X_f(G)), \text{  } v_\sigma = \sum_{i \in V(X_g(G))} v_{\sigma \wedge (i \rightarrow i')}$
\end{lemma}
\begin{proof}
Using the fact about fourier functions that $\hat{f}(x) + \hat{g}(x) = \widehat{f+g}(x)$ for all $f$ and $g$, \[
\sum_{i' \in V(X_g(G))} v_{\sigma \wedge i \rightarrow i'} = \sum_{S \in \mathbb{L}^{r/3}} \bigg(\widehat{\sum_{i' \in V(X_g(G))}h_{\sigma \wedge i \rightarrow i'}\bigg)} (S)\gamma(S)e_{[S]} = \]\[\sum_{S \in \mathbb{L}^{r/3}} ( \hat{h_{\sigma}}) (S)\gamma(S)e_{[S]} = v_\sigma\].
\end{proof}

\section{$X_f(G)$ is an expander graph}
\noindent \textbf{Lemma~\ref{lemma:CFI-expansion}} (Restated)
For any 3-regular graph $G$ even parity function $f: E(G) \rightarrow \{0, 1\}$, $Ex(X_f(G)) \geq \frac{1}{54} Ex(G)$.  Moreover if $f$ is odd, and $Ex(G) > 987/n$, where $n$ is the number of vertices in $G$, then $Ex(X_f(G)) \geq \frac{1}{54} Ex(G)$.

%\begin{theorem}
%For any 3-regular graph $G$ even parity function $f: E(G) \rightarrow \{0, 1\}$, $Ex(X_f(G)) \geq \frac{1}{12} Ex(G)$.  Moreover if $f$ is odd, and $Ex(G) > 987/n$, where $n$ is the number of vertices in $G$, then $Ex(X_f(G)) \geq \frac{1}{12} Ex(G)$.
%\footnote{We did not optimize constants.}
%\end{theorem}

To prove this theorem, we define the $k$-clustering of a graph which can be though of as a graph where each vertex $u$  is
replaced with a clique of size $deg(u)$, and each of the vertices in the clique is connected to one other node in a clique corresponding to a neighbor of $u$.

\begin{definition}\label{def:clusterign}
The \emph{clustering} of an undirected graph $G$, is a graph which we denote $Cl(G)$  which has two vertices $(u, v)$ and $(v, u)$ for each edge $(u, v) \in E(G)$.  And $E(Cl(G)) = \{((u, v), (v, u)) \mid (u, v) \in E(G)\} \cup \{((u, v_1), (u, v_2)) \mid (u, v_1), (u, v_2) \in E(G)\}$.
%A graph $H$ is said to be a \emph{clustering} of an undirected graph $G$ if the vertices of $H$ can each be labeled with an ordered pair $(u, v)$ of vertices in
%$G$ such that $(u, v)$ is an (undirected) edge in G, and furthermore $E(H) = \{((u, v), (v, u)) \mid (u, v) \in E(G)\} \cup \{((u, x), (u, y)) \mid
%(u, x) \cup (u, y) \in E(G)\}$.
\end{definition}

Additionally, we recall Definition~\ref{def:k,t-stretch} that $H$ is a $(k, t)$-stretching of a graph $G$ if it can be obtained from $G$ by  inserting at most $k$ vertices into each edge in such a way for any particular vertex, at most $t$ vertices are inserted into its incident edges.

%A k-stretching of a graph G can be thought of as a graph where at most k vertices are ``inserted'' as midpoints to each edge in G.

We now present two Lemmas that we will use in the proof of the main theorem.  We defer their proofs until later.

\begin{lemma} \label{lemma:clustering-expantion}
Let $G$ be a graph with min-degree at least 3 and max-degree $s$, then $Ex(Cl(G)) \geq \frac{1}{s} Ex(G)$.
\end{lemma}

\begin{lemma}\label{lemma:k-stretch-expantion}
Let $G$ be a graph with maximum degree $s$, and let $H$ be a $(k, t)$-stretch of $G$.
Then $Ex(H)\geq \min\{\frac{2}{sk} Ex(G),\frac{1}{t+1} Ex(G)\}$.
\end{lemma}

%May want to put a footnote explaining that the min-degree condition can be dropped, but proving it without
%that condition introduces some technical hurdles, such as the fact that 7.3.1 isn't valid in this case,
%that we don't need to tackle for the purposes of our problem.

We are now ready to prove Theorem~\ref{theorem:CFI-expansion}.

\begin{proof}
For now, assume that $f \equiv {\bf 0}$, we will later show how to get rid of this assumption.

%Next step: divide and conquer
We partition the vertices of $X_f(G)$ into the sets $S_0$ and $S_1$.  $$S_0 = \{u_{0, 0, 0} \mid u \in V(G)\} \cup \{(u, v)_0 \mid (u, v) \in E(G)\}$$
and $S_1 = V(X(G)) \setminus S_0$, and let $X_f^0(G)$ be the induced subgraph of $X_f(G)$ on $S_0$ and $X_f^1(G)$ be the
induced subgraph of $X_f(G)$ on $S_1$.\\

We now observe that both $X_f^0(G)$ and $X_f^1(G)$ have special properties.

First, $X_f^0(G)$ is a (2, 6)-stretching of $G$.  Observe that the function $g: V(G) \rightarrow X_f^0(G)$ such that $g(u) = u_{0, 0, 0}$ is a witness to this fact.
Thus by Lemma~\ref{lemma:k-stretch-expantion} we know that $Ex(X_f^0) \geq \frac{Ex(G)}{7}$

Second, $X_f^1(G)$ is a (1, 2)-stretching of $Cl(G)$, the clustering graph of $G$.  Observe that the function $g((u, v)) = (u, v)_1$ witnesses this fact, so that the clique edges of the $Cl(G)$  have an added midpoint, but the edges connecting cliques do not.  Thus by Lemma~\ref{lemma:clustering-expantion} and Lemma~\ref{lemma:k-stretch-expantion} we see that  $Ex(X_f^1(G)) \geq \frac{Ex(G)}{9}$.\\

We also note that $\frac{3}{4}$ of the vertices in $S_0$ have an edge to a vertex in $S_1$, and that $\frac{1}{2}$ of the vertices in $S_1$ have an edge to a vertex in $S_0$.  To see this note that for each $u \in G$ with neighbors $v_1, v_2, v_3$, there are 3 edges in $X(u)$ with one endpoint in $S_0$ and one in $S_1$--\\$\{((u, v_1)_0, u_{0, 1, 1}), ((u, v_2)_0, u_{1, 0, 1}), ((u, v_3)_0, u_{1, 1, 0})\}$--and each connects two different nodes.   \\

%note that no vertex in $S_0$ or $S_1$
%%TODO: maybe this is an indicator that I need better notation?
%has 2 or more edges to vertices in the other. This implies that the fraction of
%vertices in $S_0$ with an edge to a vertex in $S_1$ is $\frac{3}{4}$ and the fraction of vertices in $S_1$ with such an edge is $\frac{1}{2}$.\\

%With these lemmas and definitions in hand, we proceed to the proof that $X_f(G)$ is an expander.

Let $T \subseteq V(X_f(G))$ be given, and let $T_0 = S_0 \cap T$ and $T_1 = S_1 \cap T$.
Let $b \in \{0, 1\}$ be such that $|T_b| \leq |T_{b \oplus 1}|$.\\

Suppose that $|T_b| \geq |T_{b \oplus 1}|/5$. Using the fact that the expansion of $X_f^b(G)$ is at least $\frac{Ex(G)}{9}$, and that $6|T_b| \geq |T|$,
$Ex(T) \geq (\frac{Ex(G)}{9})/6 = \frac{Ex(G)}{54}$.\\

Suppose that  $|T_b| \leq |T_{b \oplus 1}|/5$, but $|T_{b \oplus 1}| \leq 3|S_{b \oplus 1}|/4$. Using the fact that the expansion of $X_f^{b \oplus 1}(G)$ is at least $\frac{Ex(G)}{9}$, $Ex(T_{b \oplus 1})
\geq \frac{Ex(G)}{27}$. Furthermore, $2|T_{b \oplus 1}| \geq |T|$, $Ex(T) \geq (\frac{Ex(G)}{27})/2 = \frac{Ex(G)}{54}$.\\

Suppose that  $|T_b| \leq |T_{b \oplus 1}|/5$, and $|T_{b \oplus 1}| \geq 3|S_{b \oplus 1}|/4$. Note that since at least $\frac{1}{2}|S_{b \oplus 1}|$ vertices have an edge to
$S_b$, and $T_{b \oplus 1} \geq 3|S_{b \oplus 1}|/4$, at least $|S_{b \oplus 1}|/4 \leq |T_{b \oplus 1}|/3$ vertices in $T_{b \oplus 1}$ have an edge to $S_b$, but $|T_b| \leq |T_{b \oplus 1}|/5$,
and thus $Ex(T, T^c) \geq |T_{b \oplus 1}|/3 - |T_{b \oplus 1}|/5 = 2|T_{b \oplus 1}|/15$. Since $2|T_{b \oplus 1}| \geq |T|$, $Ex(T) \geq \frac{1}{15}$, and since
$Ex(G) \leq 3$ since G is 3-regular, $Ex(T) \geq \frac{Ex(G)}{45}$.\\

Thus regardless of the choice of $T \subseteq V(X_f(G))$, $Ex(T) \geq \frac{Ex(G)}{54}$ and so $Ex(X_f(G)) \geq \frac{Ex(G)}{54}$.

It remains to relax the assumption that $f$ is the zero function.  For any $f$ with even parity $X_f(G) \cong X_{\bf 0}(G)$.  By the same reasoning, we need only show that $X_f(G)$ where the first edge is 1, is also an expander.

%\aaron{Can't we just put the next few paragraphs in the appendix, or even delete altogether? Remember, it's OK if only one of the graphs we're comparing is an expander.}

Let $f$ have odd parity.  Fix a set $S \subseteq V(X_f(G))$ containing at most half the vertices.  We will show that the expansion of $S$ is large.

First consider the case that there exists an edge $(u, v) \in E(G)$ such that for $b \in \{0, 1\}$ neither $(u, v)_b$ nor $(v, u)_b$ are in $S$.  Then by Lemma~\ref{lemma:non-iso} $V(X_f(G))$ is isomorphic to a graph $V(X_g(G))$ where $(u, v)$ is the only twisted edge.  Let $\pi$ be the isomorphism between these graphs.  Then the expansion of $S$ is identical to the expansion of $\pi(S)$.  But $V(X_g(G))$ is identical to $V(X_\mathbf{0}(G))$ except for edges between these four vertices--none of which are in $\pi(S)$.  Thus its expansion is identical to the expansion of $\pi(S)$ in $V(X_\mathbf{0}(G))$ and is at least $Ex(G)/54$.

Next consider the case that there is no such edge  $(u, v) \in E(G)$.   This mean that $S$ contains at least one vertex from each of these $3n/2$ sets, so must be of size $3n/2$.  Note that $|V(X_f(G))| =  13n$.  Thus $|S| \geq 3|V(X_f(G))|/26$.   Then by Lemma~\ref{lemma:non-iso} $V(X_f(G))$  $V(X_f(G))$ is isomorphic to a graph $V(X_g(G))$ with only one only twisted edge.  Let $\pi$ be the isomorphism between these graphs.  Then the expansion of $S$ is identical to the expansion of $\pi(S)$.  But $V(X_g(G))$ is identical to $V(X_\mathbf{0}(G))$ except for 2 edges.  Thus the number of edges leaving $\pi(S)$ in $V(X_g(G))$ is at least the number of edges leaving $\pi(S)$ in $X_\mathbf{0}(G))$ minus 2.  Thus the expansion of $S$ in $V(X_f(G))$ is at least $\frac{E(S,V \setminus S) -2}{|S|} = \frac{Ex(S)|S|-2}{|S|} = Ex(S) - \frac{2}{|S|} > Ex(S) - \frac{4}{3n}$.  Because $Ex(S) \geq \frac{47}{555}Ex(G)$ as long as $Ex{G} > 987/n$  we have that the expansion of $S$ in $V(X_f(G))$ is greater than $Ex(G)/12$.

%
%
%
%Using $\Cref{iso}$, if $f$ is even, $X_f(G) \equiv X_{\bf 0}(G)$, and isomorphic
%graphs have the same expansion. For odd $f$, $X_f(G) \cong X_{{\bf 1}_{(u, v)}}(G)$, and it is not hard to show that
%the expansion of $X_{{\bf 1}_{(u, v)}}(G)$ is neglegibly different from that of $X_{\bf 0}(G)$.
%\grant do this.

\end{proof}

We now prove Lemma~\ref{lemma:clustering-expantion}.

\begin{proof}
%We do this by showing that $Ex(G) \leq s Ex(Cl(G))$.
Let $S \subseteq V(Cl(G))$ such that $|S| \leq V(Cl(G))/2$ with minimum expansion be given.

If  $Ex(S) \geq 1$, then because $Ex(G) \leq s$, $Ex(Cl(G)) \geq \frac{1}{s} Ex(G)$ follows.

%
%We analysis two cases: 1)$Ex(S) \geq 1$, 2) $Ex(S) < 1$.
%
%Case 1: is that $Ex(S) \geq 1$.  In this case we note that the set $C(u) := \{(u, x) \mid (u, x) \in E(G)\}$ for any u has expansion 1.  Thus by minimality of $S$, $Ex(S) \leq 1$.  By assumption that $Ex(S) \geq 1$ we conclude that $Ex (S) = 1$.  Note that $Ex(G) \leq s$, and so if $Ex(H) = 1$, the conclusion $Ex(H) \geq \frac{1}{s} Ex(G)$ follows.
%
%Case 2: $Ex(S) < 1$.  For $u \in V(G)$ let $C(u) := \{(u, v) \mid (u, v) \in E(G)\}$.

Thus, we now assume that $Ex(S) < 1$.  Let $C(u) := \{(u, v) \mid (u, v) \in E(G)\}$.

\begin{claim}
If $S$ is a set of minimal expansion and $Ex(S) < 1$, then for all $u \in V(G)$ either $C(u) \cap S = C(u)$ or $C(u) \cap S = \emptyset$.
\end{claim}
\begin{proof}

Suppose that there exists a $u \in V(G)$ such that there exists $x$, $y$ such that $(u, x) \in S$ but $(u, y) \notin S$.

\begin{itemize}
\item[Case 1:] Suppose that $0 < |C(u) \cap S| \leq |C(u)|-2$.  We will show that the set $S' = S \setminus C(u)$ has less expansion than $S$, contradicting the minimal expansion property of $S$.  At most $|C(u) \cap S|$ edges are in $E(S', S'^c)$ that are not in $E(S, S^c)$ (the outgoing edge from $C(u)$ adjacent to each vertex in $C(u) \cap S$);
but there are at least $2|C(u) \cap S|$ edges in $E(S, S^c)$ that are not in $E(S', S'^c)$ (the edges within $C(u)$).
Thus $E(S', S'^c) \leq E(S, S^c) - |C(u) \cap S|$.  Giving us
$$Ex(S') \leq \frac{E(S, S^c) - |C(u) \cap S|}{|S| - |C(u) \cap S|} < \frac{E(S, S^c)}{|S|} = Ex(S),$$ the second inequality holding because $Ex(S) < 1$.

\item[Case 2:] Suppose that $C(u) \cap S = C(u) \setminus (u, v)$ for some $v$.  Let $T = S \cup \{(u, v)\}$.  If $|T| \leq n/2$ sets $S' = T$.  If $|T| > n/2$ set $S' = V(Cl(G)) \setminus T$.  We will show that the set $S'$ has less expansion than $S$, contradicting the minimal expansion property of $S$.  Either way we define $S'$, the only edge that can be in $E(S', S'^c) \setminus E(S, S^c)$ is the edge $((u, x), (x, u))$, but
the set $ E(S, S^c) \setminus E(S', S'^c)$ contains $\{((u, v_1), (u, v_2) \mid v_1 \neq v_2, (u, v_1), (u, v_2) \in E(G)\}$ whose cardinality is at least 2 since
$G$ has min-degree 3, so $E(S', S'^c) \leq E(S, S) - 1$.  We note that $|S'| \geq |S| - 1$.  Thus $$Ex(S') =  \frac{E(S', S'^c)}{|S'|} \leq \frac{E(S, S^c) - 1}{|S| - 1} < \frac{E(S, S^c)}{|S|} = Ex(S),$$ the third inequality holding because $Ex(S) < 1$.
\end{itemize}
%
%Since one of these cases must be true if there exists $u \in V(G)$ such that there exists x, y such that $(u, x) \in S$ but $(u, y) \notin S$,
%we conclude that such a u must not exist, and therefore either $C(u) \cap S = C(u)$ or $C(u) \cap S = \emptyset$.
\end{proof}

Given the above claim, we can provide a subset $T \subseteq V(G)$ such that $Ex(T) \leq s Ex(S)$, which will prove the Lemma.  Let $T = \{u \mid C(u) \cap S = C(u)\}$.
We note that $|E(T, T^c)| = |E(S, S^c)|$ since $((u, v), (v, u)) \in E(S, S^c)$ if and only if $(u, v) \in E(T, T^c)$ and by the Claim, no intracluster edges exit $S$.  %, and Lemma 7.3.1 eliminates all other possible cases of edges in $E(S, S^c)$.
Furthermore $|T| \geq \frac{1}{s}|S|$, and $|V(G) \setminus |T| \geq \frac{1}{s}|S|$.  The former is because each cluster has at most $s$ nodes. For the same reason, $|V(G) \setminus T| \geq \frac{1}{s} |V(Cl(G)) \setminus S|$.  However, because $S$ has at most have the nodes we know $|V(CL(G)) \setminus S| \geq |S|$.  Putting these together we get  $|V(G) \setminus T| \geq  \frac{1}{s} |S|$.
Thus $$Ex(G) \leq Ex(T) = \frac{|E(T, T^c)|}{\min{|T|, |V(G) \setminus T|} \leq  \frac{|E(S, S^c)|}{\frac{1}{s}|S|}} = s Ex(S) = Ex(Cl(G)),$$ as claimed.
\end{proof}

We now prove Lemma \ref{lemma:k-stretch-expantion}.
\begin{proof}
%We do this by showing that $Ex(G) \leq sk Ex(H)$.

Let $S \subseteq V(H) \mid |S| \leq V(H)/2$ with minimum expansion be given.

If  $Ex(S) \geq \frac{2}{k}$, then because $Ex(G) \leq s$, $Ex(H) \geq \frac{2}{ks} Ex(G)$ follows.

%
%We analysis two cases: 1)$Ex(S) \geq 1$, 2) $Ex(S) < 1$.
%
%Case 1: is that $Ex(S) \geq 1$.  In this case we note that the set $C(u) := \{(u, x) \mid (u, x) \in E(G)\}$ for any u has expansion 1.  Thus by minimality of $S$, $Ex(S) \leq 1$.  By assumption that $Ex(S) \geq 1$ we conclude that $Ex (S) = 1$.  Note that $Ex(G) \leq s$, and so if $Ex(H) = 1$, the conclusion $Ex(H) \geq \frac{1}{s} Ex(G)$ follows.
%
%Case 2: $Ex(S) < 1$.  For $u \in V(G)$ let $C(u) := \{(u, v) \mid (u, v) \in E(G)\}$.

Thus, we now assume that $Ex(S) < \frac{2}{k}$.  Let $g$ witness the fact that $H$ is a $(k, t)$-stretching of $G$.

%Let $S \subseteq V(H)$ such that $|S| \leq V(H)/2$ with minimum expansion be given. Since $H$ is a $k$-stretch of $G$, let $g$ be the function
%witnessing the fact that $H$ is a $k$-stretching of $G$.\\
%
%Observe that $E(H) \leq \frac{2}{k}$, as each path, excluding the ends, has this expansion.
%If $E(H) = \frac{2}{k}$, then since all s-regular graphs have expansion at most s,
%the conclusion $E(G) \leq sk E(H)$ follows.
%Thus for the remainder of this analysis, we assume that $E(H) < \frac{2}{k}$.

\begin{claim}  \label{claim:subclaim}
If $S$ is a set of minimal expansion and $Ex(S) < \frac{2}{k}$, then for each $(u, v) \in E(G)$:
\begin{itemize}
\item $g(u), g(v) \not\in S$, in which case none of $P(u, v)$ is in $S$.
\item $g(u) \in S$ but $g(v) \not\in S$, in which case there exists $i^*$ such that $p(u, v)_i \in S$ for $0 \leq i \leq i^*$ and $p(u, v)_i \not\in S$ for $i^* < i \leq \ell$; similarly for $g(v) \in S$ but $g(u) \not\in S$.
\item $g(u), g(v) \in S$ in which case all of $P(u, v) \setminus g(V(G))$ is in $S$.

\end{itemize}
%$P(u, v)$ has at most 1 edge in $E(S, S^c)$.
\end{claim}
\begin{proof}
Consider the case where  $g(u), g(v) \not\in S$.  We claim that no vertex of $P(u, v) \setminus g(V(G))$ is in $S$.  If some were, then the set $S' = S \setminus P(u, v)$ would have less expansion then $S$ violating its minimality property.  There is at least two edges which go from $S \cup P(u, v)$ to $P(u, v) \setminus S$.  But $S'$ has no edges which goes from $S' \cup P(u, v)$ to $P(u, v) \setminus S$. Because $S$ and $S'$ differ only on $P(u, v) \setminus g(V(G))$, we know that $S'$ has at least two fewer edges leaving it than $S$.  But $S'$ also has at least $|S| - k$ vertices.   Thus $Ex(S') = \frac{|E(S', S'^c)|}{|S'|} < \frac{|E(S, S^c)| - 2}{|S| - k} \leq Ex(S)$.  The penultimate inequality follows because  $Ex(S) < \frac{2}{k}$.

Say that  $g(u) \in S$ and $g(v) \not\in S$, but there is no $i^*$ such that $p(u, v)_i \in S$ for $1 \leq i \leq i^*$ and $p(u, v)_i \not\in S$ for $i^* < i \leq \ell$.  Then $|E(S \cap P(u, v), P(u, v) \setminus S)| > 2$.  Let $j^*$ be the number of vertices in $P(u, v) \cap S$.  The set $S' = (S \setminus P(u, v)) \cup  \{p(u, v)_i\}_{0 \leq i < j^*}$ will have less expansion then $S$ violating its minimality property.  Note that $|E(S' \cap P(u, v), P(u, v) \setminus S')| =1$.  Because $S$ and $S'$ differ only on $P(u, v) \setminus g(V(G))$, we know that $S'$ has at least one fewer edges leaving it than $S$.  However, because $|S| = |S'|$ we see that $Ex(S') < Ex(S)$.

The case where  $g(v) \in S$ and $g(u) \not\in S$ is symmetric.

Finally, consider first the case that $g(u), g(v) \in S$, but some vertices of $P(u,v)$ are not in $S$.  Then the set $S' = S \cup P(u, v)$ would have less expansion than $S$ violating its minimality property.  Note that $|E(S \cap P(u, v), P(u, v) \setminus S)| \geq 2$ while $|E(S' \cap P(u, v), P(u, v) \setminus S')| =0$.  Because $S$ and $S'$ differ only on $P(u, v) \setminus g(V(G))$, we know that $S'$ has at least two fewer edges leaving it than $S$.

However, $|S| - k \geq \min\{|S'|, |V(H) \setminus S'|$.  Thus $Ex(S) = \frac{E(S, V(H)- S)}{|S|} \geq \frac{E(S, V(H)- S)-2}{|S|-k} > \frac{E(S', V(H)- S')}{ \min\{|S'|, |V(H) \setminus S'|} = Ex(S')$.  The penultimate inequality follows because  $Ex(S) < \frac{2}{k}$.
%
%
%Consider $S' = S \cup P(u, v)$.  Then adding these vertices will strictly increase the number of vertices in $S$, but will not increase number of edges leaving $S$ (the only neighbors these vertices have are in this path).  Thus $Ex(S') < Ex(S)$.
%
%
%
%
%Consider second the case that $g(v) \not\in S$ as well.  Then there is at least one edge which goes from $S \cup P(u, v)$ to $P(u, v) \setminus S$.  But $S'$ has at most one edge which goes from $S' \cup P(u, v)$ to $P(u, v) \setminus S$. Because $S$ and $S'$ differ only on $P(u, v) \setminus g(V(G))$, we know that $S'$ has fewer edges leaving it than $S$.  But $S'$ also has strictly more vertices than $S$.   Thus $Ex(S') < Ex(S)$.
%
%The argument where $g(v) \in S$ but $g(u) \not\in S$ is analogous.
%

\end{proof}

Given the above claim, we can provide a subset $T \subseteq V(G)$ such that $Ex(T) \leq (t + 1) Ex(S)$, which will prove the Lemma.  Let $T = g^{-1}(S)$.  By Claim \ref{claim:subclaim}, $|E(T, T^c)| = |E(S, S^c)|$.  Furthermore $|T| \geq \frac{1}{t+1}|S|$, and $|V(G) \setminus |T| \geq \frac{1}{t}|S|$.  The former is because each vertex of $G$ corresponds to at most $t+1$ vertices in $H$.  For the same reason, $|V(G) \setminus T| \geq \frac{1}{t+1} |V(H) \setminus S|$.  However, because $S$ has at most half the nodes we know $|V(H) \setminus S| \geq |S|$.  Putting these together we get  $|V(G) \setminus T| \geq  \frac{1}{t+1} |S|$.
\end{proof}

%\section{Conclusion}
%\input{conclusion}

%\section{Resources}

\bibliographystyle{abbrv}
\bibliography{grant}

\begin{thebibliography}{10}

\bibitem{AroraG11}
S.~Arora and R.~Ge.
\newblock New tools for graph coloring.
\newblock In {\em 14th International Workshop, APPROX 2011, and 15th
  International Workshop, RANDOM 2011}, volume 6845 of {\em Lecture Notes in
  Computer Science}, pages 1--12. Springer, 2011.

\bibitem{AroraKKSTV08}
S.~Arora, S.~Khot, A.~Kolla, D.~Steurer, M.~Tulsiani, and N.~K. Vishnoi.
\newblock Unique games on expanding constraint graphs are easy: extended
  abstract.
\newblock In {\em Proceedings of the 40th ACM Symposium on Theory of
  Computing}, pages 21--28, 2008.

\bibitem{AtseriasM12}
A.~Atserias and E.~Maneva.
\newblock Sherali-adams relaxations and indistinguishability in counting
  logics.
\newblock In {\em Proceedings of the 3rd Innovations in Theoretical Computer
  Science Conference}, ITCS '12, pages 367--379, New York, NY, USA, 2012. ACM.

\bibitem{Babai79}
L.~Babai.
\newblock Monte {C}arlo algorithms in graph isomorphism testing.
\newblock In {\em Universit\'e de Montr\'eal Technical Report}, volume~70 of
  {\em DMS}, page~42, 1979.
\newblock {\tt http://people.cs.uchicago.edu/$\sim$laci/lasvegas79.pdf}.

\bibitem{BabaiCSSW13}
L.~Babai, X.~Chen, X.~Sun, S.-H. Teng, and J.~Wilmes.
\newblock Faster canonical forms for strongly regular graphs.
\newblock In {\em Proceedings of the 54th Annual Symposium on Foundations of
  Computer Science (Berkeley, CA, USA), FOCS}, volume~13, 2013.

\bibitem{BabaiL83}
L.~Babai and E.~M. Luks.
\newblock Canonical labeling of graphs.
\newblock In {\em Proc. 15th STOC}, pages 171--183. ACM Press, 1983.

\bibitem{BabaiM88}
L.~Babai and S.~Moran.
\newblock Arthur-{M}erlin games: a randomized proof system, and a hierarchy of
  complexity classes.
\newblock {\em J. Computer and Sys. Sci.}, 36:254--276, 1988.

\bibitem{BarakRS11}
B.~Barak, P.~Raghavendra, and D.~Steurer.
\newblock Rounding semidefinite programming hierarchies via global correlation.
\newblock In {\em Proceedings of the 52th IEEE Symposium on Foundations of
  Computer Science}, 2011.

\bibitem{BoppanaHZ87}
R.~B. Boppana, J.~Hastad, and S.~Zachos.
\newblock Does co-{N}{P} have short interactive proofs?
\newblock {\em Information Processing Letters}, 25:127--132, May 1987.

\bibitem{CharikarMM-07}
M.~Charikar, K.~Makarychev, and Y.~Makarychev.
\newblock Integrality gaps for {S}herali-{A}dams relaxations.
\newblock In {\em Proceedings of the 41st ACM Symposium on Theory of
  Computing}, 2009.

\bibitem{Ellis}
D.~Ellis.
\newblock The expansion of random regular graphs.
\newblock https://www.dpmms.cam.ac.uk/\~dce27/randomreggraphs3.pdf.

\bibitem{GolreichMW91}
O.~Goldreich, S.~Micali, and A.~Wigderson.
\newblock Proofs that yield nothing but their validity, or all languages in
  {N}{P} have zero-knowledge proof systems.
\newblock {\em Journal ACM}, 38:690--728, 1991.

\bibitem{GuruswamiS11}
V.~Guruswami and A.~K. Sinop.
\newblock Lasserre hierarchy, higher eigenvalues, and approximation schemes for
  graph partitioning and quadratic integer programming with psd objectives.
\newblock In {\em 52nd Annual Symposium on Foundations of Computer Science},
  pages 482--491. IEEE, 2011.

\bibitem{GuruswamiS13}
V.~Guruswami and A.~K. Sinop.
\newblock Approximating non-uniform sparsest cut via generalized spectra.
\newblock In {\em Proceedings of the Twenty-Fourth Annual ACM-SIAM Symposium on
  Discrete Algorithms}, pages 295--305. SIAM, 2013.

\bibitem{bliss}
T.~Junttila and P.~Kaski.
\newblock Conflict propagation and component recursion for canonical labeling.
\newblock In {\em Proc. First ICST}, TAPAS, pages 151--162, 2011.

\bibitem{saucy}
H.~Katebi, K.~A. Sakallah, and I.~L. Markov.
\newblock Conflict {A}nticipation in the {S}earch for {G}raph {A}utomorphisms.
\newblock In {\em Int'l Conf. on Logic for Programming, Artificial Intelligence
  and Reasoning (LPAR)}, Merida, Venezuela, 2012.

\bibitem{Luks82}
E.~M. Luks.
\newblock Isomorphism of graphs of bounded valence can be tested in polynomial
  time.
\newblock {\em J. Comp. Sys. Sci.}, 25:42--65, 1982.

\bibitem{nauty}
B.~D. McKay.
\newblock Practical {G}raph {I}somorphism.
\newblock In {\em Congressus Numerantium, 30}, pages 45--87, 1981.

\bibitem{MontanariS09}
A.~Montanari and A.~Saberi.
\newblock Convergence to equilibrium in local interaction games.
\newblock In {\em Proceedings of the 50th IEEE Symposium on Foundations of
  Computer Science}, 2009.

\bibitem{ODonnell2013}
R.~O'Donnell.
\newblock Analysis of boolean functions, 2013.
\newblock http://www.contrib.andrew.cmu.edu/~ryanod/.

\bibitem{ODonnellWWZ2014}
R.~O'Donnell, J.~Wright, C.~Wu, and Y.~Zhou.
\newblock Hardness of robust graph isomorphism, lasserre gaps, and asymmetry of
  random graphs.
\newblock In {\em Proceedings of the Twenty-Fourth Annual ACM-SIAM Symposium on
  Discrete Algorithms}, 2014.
\newblock To Appear.

\bibitem{S08}
G.~Schoenebeck.
\newblock Linear level {L}asserre lower bounds for certain k-{CSP}s.
\newblock In {\em Proceedings of the 49th IEEE Symposium on Foundations of
  Computer Science}, pages 593--692, 2008.

\bibitem{STT-07b}
G.~Schoenebeck, L.~Trevisan, and M.~Tulsiani.
\newblock A linear round lower bound for {L}ovasz-{S}chrijver {SDP} relaxations
  of {V}ertex {C}over.
\newblock In {\em Proceedings of the 39th ACM Symposium on Theory of Computing
  (STOC07)}, 2007.
\newblock Earlier version appeared as Technical Report TR06-098 on Electronic
  Colloquium on Computational Complexity.

\bibitem{STT-2007a}
G.~Schoenebeck, L.~Trevisan, and M.~Tulsiani.
\newblock Tight integrality gaps for {L}ovasz-{S}chrijver {LP} relaxations of
  {V}ertex {C}over and {M}ax {C}ut.
\newblock In {\em Proceedings of the 39th ACM Symposium on Theory of
  Computing}, 2007.
\newblock Earlier version appeared as Technical Report TR06-132 on Electronic
  Colloquium on Computational Complexity.

\bibitem{Spielman96}
D.~A. Spielman.
\newblock Faster isomorphism testing of strongly regular graphs.
\newblock In {\em Proceedings of the 28th ACM Symposium on Theory of
  Computing}, pages 576--584, New York, NY, USA, 1996. ACM.

\bibitem{WeisfeilerL68}
B.~Weisfeiler and A.~A. Lehman.
\newblock A reduction of a graph to a canonical form and an algebra arising
  during this reduction, (in russian).
\newblock {\em Nauchno-Technicheskaya Informatsia}, Seriya 2, 9:12--16, 1968.

\bibitem{CaiFI89}
J.~yi~Cai and N.~Martin~F\"urer, Neil~Immerman.
\newblock An optimal lower bound on the number of variables for graph
  identification.
\newblock {\em Proceedings of the 30th IEEE Symposium on Foundations of
  Computer Science}, 12(4):389--410, 1989.

\end{thebibliography}

\renewcommand{\thesection}{\Alph{section}}
\setcounter{section}{0}

%\section{Appendix}
\section{Additional Background on \GI} \label{appendix:background}

\paragraph{Partition and Refinement}\label{sec:Partition-Refinement}
One of the first approaches to the (colored) \GI problem is \emph{partition and refinement}. The initial coloring of colored graphs induces a partition of the set of vertices into sets with the same color. The procedure iteratively refines this partition, by assigning new colors to the vertices. The new color of a vertex is the set of colors of its neighbors. The refinement step can be repeated until the coloring converges. If after any step the two graphs have a different number of vertices of a certain color, then the graphs are not isomorphic.  The procedure can be applied to uncolored graphs as well: initially assign the same color to all vertices, the first step of refinement will label vertices by their degree.  Let $n$ be the number of vertices of the graphs. At every step, there are at most $n$ distinct labels, and there are at most $n$ refinement steps before convergence, hence this procedure runs in polynomial-time. Every isomorphism that preserves the initial coloring must preserve all refinements obtained in this procedure. However, there is no guarantee that when the algorithm converges the two graphs are indeed isomorphic. For example, for uncolored regular graphs of the same degree, the procedure does not even get started. This procedure was first introduced by Weisfeiler and Lehman~\cite{WeisfeilerL68}

\paragraph{$k$-WL}
A natural generalization of this algorithm looks at more than neighbors for refinement. The $k$-dimensional Weisfeler-Lehman process ($k$-WL from now on) generalizes the refinement step described above. The new color of a vertex $v$ is given by the set of (colored) isomorphism types of the subgraphs of size $k$ that include $v$. We refer the reader to~\cite{CaiFI89} for a more complete description of this process, as it is not required to understand this paper. Since it must look at all subsets of vertices of size $k$, the new algorithm runs in time $n^{O(k)}$. There was hope that $k$-WL for small values of $k$ might solve \GI. This was shown not to be the case by Cai, F\"urer, and Immerman~\cite{CaiFI89}. They constructed families of non-isomorphic pairs of graphs that fail to be distinguished by the $k$-WL algorithm for any $k = \Omega(n)$, thus ruling out a subexponential running time.    Recently, Atserias and Maneva showed that $k$-WL is equivalent to $k\pm 1$ rounds of the Sherali-Adams hierarchy~\cite{AtseriasM12}. We extend their result to the Lasserre hierarchy, while relaxing the constants.

\paragraph{Group Theoretic Methods for \GI}
The approach that yields the fastest (worst-case) algorithms for \GI is group theoretic. These algorithms exploit the group-theoretic structure of the space of permutations and the set of isomorphisms\footnote{the first is a group, the second a coset of the subgroup of automorphisms, isomorphisms of the graph to itself}. This approach was first introduced by Babai~\cite{Babai79}. Luks used group theory in greater depth to obtain a polynomial-time algorithm for graphs of bounded degree~\cite{Luks82}. Combined with a combinatorial trick due to Zemlyachenko, Luks's algorithm yielded an algorithm for \GI that runs in time $\exp(\sqrt{n\log(n)})$ (see \cite{BabaiL83}), which remains the best known run-time for the general problem.

It is noteworthy that group theory was first used in the \GI problem to obtain a polynomial-time algorithm for colored graphs with bounded color-class size~\cite{Babai79}. The graphs constructed by Cai, F\"urer, and Immermann~\cite{CaiFI89} fall into this class, and thus are decidable in polynomial time even by this first application of group theory (though they were constructed a decade later).

\paragraph{\GI in practice}
\GI is of practical interest, and several software packages (see e.\,g.\, Nauty~\cite{nauty}, Saucy~\cite{saucy} and Bliss \cite{bliss}) use a combination of partition and refinement, backtrack search, heuristics, and some basic group theory to solve the problem.  Unlike $k$-WL, which may or many not get a correct answer, but has a guaranteed running time; these algorithms guarantee a correct answer, but have no guarantees on running time.  Experimental evaluation shows these software packages work well on a large variety of instances, including the particular Cai, F\"urer, and Immerman instance that shows the failure of $k$-WL.  Miazaki generalized this family of instances to be hard in practice.  With  repeated refinements, the software now can solve many these, at least when the instances are moderately sized (note that even a cubic running time would be prohibitive for graphs with on the order of $10^4$ vertices).

\section{Picture of $X_f(G)$} \label{appendix:pictures}
Suppose that G is the following graph:

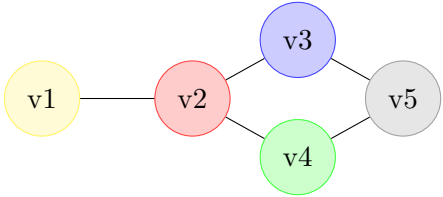
\vspace{10mm}

\begin{tikzpicture}[yscale=0.6]

  \node[circle,draw=red!75,fill=red!20,minimum size=10mm] (v2) at (0, 3.5) {v2};
  \node[circle,draw=yellow!75,fill=yellow!20,minimum size=10mm] (v1) at (-2.0, 3.5) {v1};
  \node[circle,draw=blue!75,fill=blue!20,minimum size=10mm] (v3) at (1.4, 4.8) {v3};
  \node[circle,draw=green!75,fill=green!20,minimum size=10mm] (v4) at (1.4, 2.2) {v4};
  \node[circle,draw=gray!75,fill=gray!20,minimum size=10mm] (v5) at (2.8, 3.5) {v5};

  \draw (v1) -- (v2);
%  \node at (-1, 3.5) {1};
  \draw (v2) -- (v3);
 % \node at (0.7, 4.3){2};
  \draw (v2) -- (v4);
 % \node at (0.7, 2.8){3};
  \draw (v3) -- (v5);
%  \node at (2.1, 4.3){4};
  \draw (v4) -- (v5);
  %\node at (2.1, 2.8){5};

\end{tikzpicture}

Then $X_{\bf 0}(G)$ is the following graph:

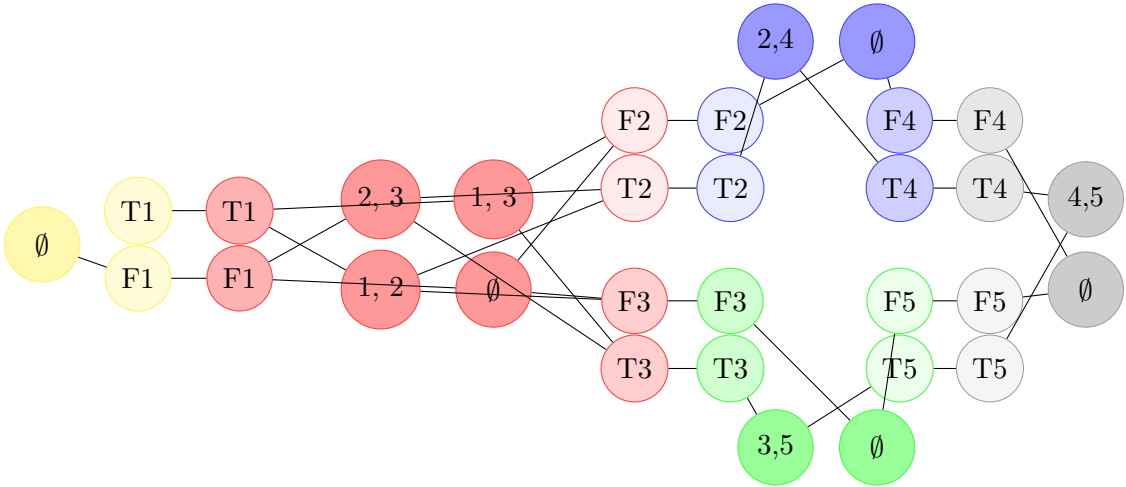
\vspace{0.5cm}

\begin{tikzpicture}[xscale=0.75, yscale=0.6]

   \tikzstyle{mid}=[circle,minimum size=10mm]
   \tikzstyle{edge}=[circle,minimum size=7mm]

  \node [mid,draw=yellow!75,fill=yellow!40] (mm) at (0, 3.5) {$\emptyset$};
  \node [edge,draw=yellow!75,fill=yellow!20] (tt1) at (1.7,4.25)  {T1};
  \node [edge,draw=yellow!75,fill=yellow!20] (ff1) at (1.7,2.75) {F1};
  \node [mid,draw=red!75,fill=red!40] (m12) at (6,2.5)  {1, 2};
  \node [mid,draw=red!75,fill=red!40] (m13) at (8,4.5)  {1, 3};
  \node [mid,draw=red!75,fill=red!40] (m23) at (6,4.5) {2, 3};
  \node [mid,draw=red!75,fill=red!40] (m) at (8,2.5)  {$\emptyset$};
  \node [edge,draw=red!75,fill=red!30] (t1) at (3.5,4.25)  {T1};
  \node [edge,draw=red!75,fill=red!30] (f1) at (3.5,2.75) {F1};
  \node [edge,draw=red!75,fill=red!19] (t3) at (10.5,0.75)  {T3};
  \node [edge,draw=red!75,fill=red!19] (f3) at (10.5,2.25) {F3};
  \node [edge,draw=red!75,fill=red!8] (t2) at (10.5,4.75)  {T2};
  \node [edge,draw=red!75,fill=red!8] (f2) at (10.5,6.25) {F2};
  \node [mid,draw=blue!75,fill=blue!40] (mmm) at (14.8,8)  {$\emptyset$};
  \node [mid,draw=blue!75,fill=blue!40] (mmm24) at (13,8)  {2,4};
  \node [edge,draw=blue!75,fill=blue!8] (tt2) at (12.2,4.75)  {T2};
  \node [edge,draw=blue!75,fill=blue!8] (ff2) at (12.2,6.25) {F2};
  \node [edge,draw=blue!75,fill=blue!19] (t4) at (15.2,4.75)  {T4};
  \node [edge,draw=blue!75,fill=blue!19] (f4) at (15.2,6.25) {F4};
  \node [mid,draw=green!75,fill=green!40] (mmmm) at (14.8,-1)  {$\emptyset$};
  \node [mid,draw=green!75,fill=green!40] (mmmm35) at (13,-1)  {3,5};
  \node [edge,draw=green!75,fill=green!19] (tt3) at (12.2,0.75)  {T3};
  \node [edge,draw=green!75,fill=green!19] (ff3) at (12.2,2.25) {F3};
  \node [edge,draw=green!75,fill=green!8] (t5) at (15.2,0.75)  {T5};
  \node [edge,draw=green!75,fill=green!8] (f5) at (15.2,2.25) {F5};
  \node [mid,draw=gray!75,fill=gray!40] (mmmmm) at (18.5,2.5)  {$\emptyset$};
  \node [mid,draw=gray!75,fill=gray!40] (mmmmm45) at (18.5,4.5)  {4,5};
  \node [edge,draw=gray!75,fill=gray!8] (tt5) at (16.8,0.75)  {T5};
  \node [edge,draw=gray!75,fill=gray!8] (ff5) at (16.8,2.25) {F5};
  \node [edge,draw=gray!75,fill=gray!19] (tt4) at (16.8,4.75)  {T4};
  \node [edge,draw=gray!75,fill=gray!19] (ff4) at (16.8,6.25) {F4};

  \foreach \from/\to in {m/f1,m/f2,m/f3,m12/t1,m12/t2,m12/f3,m13/t1,m13/f2,m13/t3,m23/f1,m23/t2,m23/t3,t1/tt1,f1/ff1,mm/ff1,
                                     mmmmm/ff5,mmmmm/ff4,f4/ff4,t4/tt4,f5/ff5,t5/tt5,t3/tt3,f3/ff3,f2/ff2,t2/tt2,mmm/ff2,mmm/f4,mmm24/tt2,mmm24/t4,
                                     mmmm/ff3,mmmm/f5,mmmm35/tt3,mmmm35/t5,mmmmm45/tt5,mmmmm45/tt4}
    \draw (\from) -- (\to);

\end{tikzpicture}

And if $f$ is specified by $f((v_1, v_2))= 1$, $f((v_3, v_5)) = 1$, and $f((u, v)) = 0$ otherwise, $X_f(G)$ is the following:

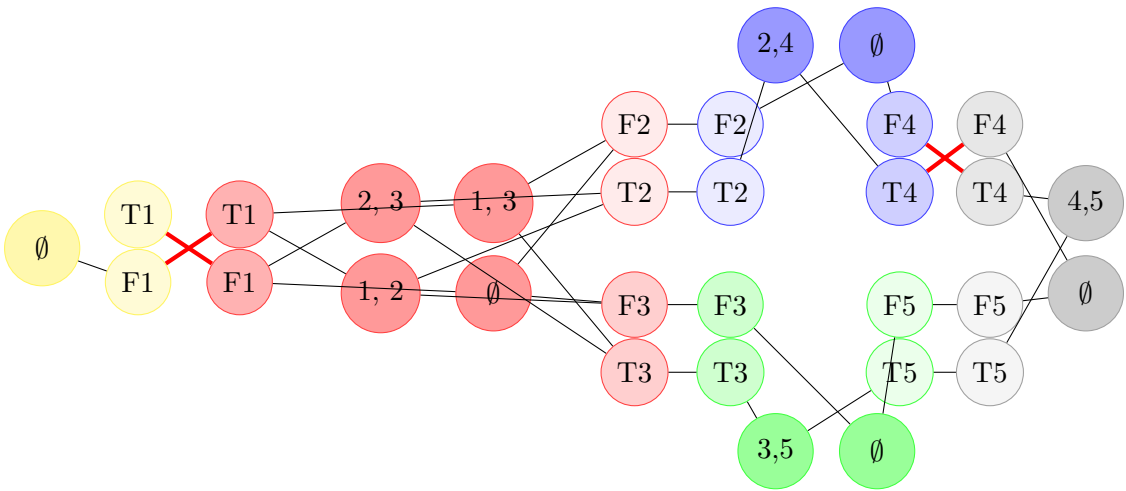
\vspace{0.5cm}

\begin{tikzpicture}[xscale=0.75, yscale=0.6]

   \tikzstyle{mid}=[circle,minimum size=10mm]
   \tikzstyle{edge}=[circle,minimum size=7mm]

  \node [mid,draw=yellow!75,fill=yellow!40] (mm) at (0, 3.5) {$\emptyset$};
  \node [edge,draw=yellow!75,fill=yellow!20] (tt1) at (1.7,4.25)  {T1};
  \node [edge,draw=yellow!75,fill=yellow!20] (ff1) at (1.7,2.75) {F1};
  \node [mid,draw=red!75,fill=red!40] (m12) at (6,2.5)  {1, 2};
  \node [mid,draw=red!75,fill=red!40] (m13) at (8,4.5)  {1, 3};
  \node [mid,draw=red!75,fill=red!40] (m23) at (6,4.5) {2, 3};
  \node [mid,draw=red!75,fill=red!40] (m) at (8,2.5)  {$\emptyset$};
  \node [edge,draw=red!75,fill=red!30] (t1) at (3.5,4.25)  {T1};
  \node [edge,draw=red!75,fill=red!30] (f1) at (3.5,2.75) {F1};
  \node [edge,draw=red!75,fill=red!19] (t3) at (10.5,0.75)  {T3};
  \node [edge,draw=red!75,fill=red!19] (f3) at (10.5,2.25) {F3};
  \node [edge,draw=red!75,fill=red!8] (t2) at (10.5,4.75)  {T2};
  \node [edge,draw=red!75,fill=red!8] (f2) at (10.5,6.25) {F2};
  \node [mid,draw=blue!75,fill=blue!40] (mmm) at (14.8,8)  {$\emptyset$};
  \node [mid,draw=blue!75,fill=blue!40] (mmm24) at (13,8)  {2,4};
  \node [edge,draw=blue!75,fill=blue!8] (tt2) at (12.2,4.75)  {T2};
  \node [edge,draw=blue!75,fill=blue!8] (ff2) at (12.2,6.25) {F2};
  \node [edge,draw=blue!75,fill=blue!19] (t4) at (15.2,4.75)  {T4};
  \node [edge,draw=blue!75,fill=blue!19] (f4) at (15.2,6.25) {F4};
  \node [mid,draw=green!75,fill=green!40] (mmmm) at (14.8,-1)  {$\emptyset$};
  \node [mid,draw=green!75,fill=green!40] (mmmm35) at (13,-1)  {3,5};
  \node [edge,draw=green!75,fill=green!19] (tt3) at (12.2,0.75)  {T3};
  \node [edge,draw=green!75,fill=green!19] (ff3) at (12.2,2.25) {F3};
  \node [edge,draw=green!75,fill=green!8] (t5) at (15.2,0.75)  {T5};
  \node [edge,draw=green!75,fill=green!8] (f5) at (15.2,2.25) {F5};
  \node [mid,draw=gray!75,fill=gray!40] (mmmmm) at (18.5,2.5)  {$\emptyset$};
  \node [mid,draw=gray!75,fill=gray!40] (mmmmm45) at (18.5,4.5)  {4,5};
  \node [edge,draw=gray!75,fill=gray!8] (tt5) at (16.8,0.75)  {T5};
  \node [edge,draw=gray!75,fill=gray!8] (ff5) at (16.8,2.25) {F5};
  \node [edge,draw=gray!75,fill=gray!19] (tt4) at (16.8,4.75)  {T4};
  \node [edge,draw=gray!75,fill=gray!19] (ff4) at (16.8,6.25) {F4};

  \foreach \from/\to in {m/f1,m/f2,m/f3,m12/t1,m12/t2,m12/f3,m13/t1,m13/f2,m13/t3,m23/f1,m23/t2,m23/t3,mm/ff1,
                                     mmmmm/ff5,mmmmm/ff4,f5/ff5,t5/tt5,t3/tt3,f3/ff3,f2/ff2,t2/tt2,mmm/ff2,mmm/f4,mmm24/tt2,mmm24/t4,
                                     mmmm/ff3,mmmm/f5,mmmm35/tt3,mmmm35/t5,mmmmm45/tt5,mmmmm45/tt4}
    \draw (\from) -- (\to);
  \draw [color=red, ultra thick] (t1) -- (ff1);
  \draw [color=red, ultra thick] (f1) -- (tt1);
  \draw [color=red, ultra thick] (t4) -- (ff4);
  \draw [color=red, ultra thick] (f4) -- (tt4);
\end{tikzpicture}

\section{Computation of the vectors for individual vertices} \label{appendix:vectorcomputation}

\begin{itemize}
\item For $(u, v) \in E(G)$,  \[v_{(u, v)_b \rightarrow (u, v)_c'} = \hat{h}_{(u, v)_b \rightarrow (u, v)_c'}(\emptyset)\gamma(\emptyset)e_{[\emptyset]} +
\hat{h}_{(u, v)_b \rightarrow (u, v)_c'}(\{(u, v)\})\gamma((u, v))e_{[(u, v)]}\] \[=
\frac{1}{2}e_{[\emptyset]} +
 \frac{1}{2}(-1)^{b \oplus c}\gamma((u, v))e_{[(u, v)]}\]

\item For $u \in V(G)$ with $\Gamma(u) = \{u_1, u_2, u_3\}$, if we denote $u_{b_{u_1}, b_{u_2} b_{u_3}} \rightarrow u'_{\bar{b}_{u_1}, \bar{b}_{u_2}, \bar{b}_{u_3}}$ by $\sigma$, and assuming that $\gamma({u, u_i}) = 1$ for all i, we have \[v_{u_{b_{u_1}, b_{u_2} b_{u_3}} \rightarrow u'_{\bar{b}_{u_1}, \bar{b}_{u_2}, \bar{b}_{u_3}}} =
\]\[
\hat{h}_{\sigma}(\emptyset)\gamma(\emptyset)e_{[\emptyset]}
+ \hat{h}_{\sigma}(\{(u, u_1)\})\gamma({(u, u_1)})e_{[\{(u, u_1)\}]}
+ \hat{h}_{\sigma}(\{(u, u_2)\})\gamma({(u, u_2)})e_{[\{(u, u_2)\}]}\]\[
+ \hat{h}_{\sigma}(\{(u, u_3)\})\gamma({(u, u_3)})e_{[\{(u, u_3)\}]}
+ \hat{h}_{\sigma}(\{(u, u_1), (u, u_2)\}) \gamma(\{(u, u_1), (u, u_2)\})e_{[\{(u, u_1), (u, u_2)\}]}\]\[
+ \hat{h}_{\sigma}(\{(u, u_1), (u, u_3)\}) \gamma({(u, u_1), (u, u_3)})e_{[\{(u, u_1), (u, u_3)\}]}
+ \hat{h}_{\sigma}(\{(u, u_2), (u, u_3)\}) \gamma({(u, u_2), (u, u_3)})e_{[\{(u, u_2), (u, u_3)\}]}\]\[
+ \hat{h}_{\sigma}(\{(u, u_1), (u, u_2), (u, u_3)\}) \gamma(\{(u, u_1), (u, u_2), (u, u_3)\}])e_{[\{(u, u_1), (u, u_2), (u, u_3)\}]}\]\[ =
\frac{1}{8}e_{[\emptyset]}
+ \frac{1}{8}(-1)^{b_{u_1} \oplus \bar{b}_{u_1}}e_{[\{(u, u_1)\}]}
+ \frac{1}{8}(-1)^{b_{u_2} \oplus \bar{b}_{u_2}}e_{[\{(u, u_2)\}]}\]\[
+ \frac{1}{8}(-1)^{b_{u_3} \oplus \bar{b}_{u_3}}e_{[\{(u, u_3)\}]}
+ \frac{1}{8}(-1)^{b_{u_1} \oplus \bar{b}_{u_1} \oplus b_{u_2} \oplus \bar{b}_{u_2}}e_{[\{(u, u_1), (u, u_2)\}]}
+ \frac{1}{8}(-1)^{b_{u_1} \oplus \bar{b}_{u_1} \oplus b_{u_3} \oplus \bar{b}_{u_3}}e_{[\{(u, u_1), (u, u_3)\}]}\]\[
+ \frac{1}{8}(-1)^{b_{u_2} \oplus \bar{b}_{u_2} \oplus b_{u_3} \oplus \bar{b}_{u_3}}e_{[\{(u, u_2), (u, u_3)\}]}
+ \frac{1}{8}(-1)^{b_{u_1} \oplus \bar{b}_{u_1} \oplus b_{u_2} \oplus \bar{b}_{u_2} \oplus b_{u_3} \oplus \bar{b}_{u_3}}e_{[\{(u, u_1), (u, u_2), (u, u_3)\}]}\]\[ =
%\frac{1}{8} v_{[\emptyset]} + \frac{1}{8} v_{[y_{(u, u_1)}, y_{(u, u_2)}, y_{(u, u_3)}]} +
%\frac{1}{8} v_{[y_{(u, u_1)}]} + \frac{1}{8} v_{[y_{(u, u_2)}, y_{(u, u_3)}]} +
%\frac{1}{8} v_{[y_{(u, u_2)}]} + \frac{1}{8} v_{[y_{(u, u_1)}, y_{(u, u_3)}]} +
%\frac{1}{8} v_{[y_{(u, u_3)}]} + \frac{1}{8} v_{[y_{(u, u_1)}, y_{(u, u_2)}]} =
\frac{1}{4} e_{[\emptyset]} + (-1)^{b_{u_1} \oplus \bar{b}_{u_1}} \frac{1}{4} e_{[(u, u_1)]} + (-1)^{b_{u_2} \oplus \bar{b}_{u_2}}
\frac{1}{4} e_{[(u, u_2)]} + (-1)^{b_{u_3 \oplus \bar{b}_{u_3}}}\frac{1}{4} e_{[(u, u_3)]}\]
\end{itemize}
The last inequality comes from the fact that $b_{u_1} \oplus \bar{b}_{u_1} \oplus b_{u_2} \oplus \bar{b}_{u_2} \oplus b_{u_3} \oplus \bar{b}_{u_3} = 0$,
and that ${[\{(u, u_3)\}]} = {[\{(u, u_1), (u, u_2)\}]}$, ${[\{(u, u_1)\}]} = {[\{(u, u_2), (u, u_3)\}]}$, ${[\{(u, u_2)\}]} = {[\{(u, u_1), (u, u_3)\}]}$,\\
and $[\emptyset] = [\{(u, u_1), (u, u_2), (u, u_3)\}]$. 

\end{document}